\RequirePackage[l2tabu,orthodox]{nag}		
\documentclass[reqno]{amsart}
\usepackage[margin=1.55in,bottom=1.25in,top=1.25in]{geometry}		

\usepackage{amsmath}		
\usepackage{amssymb}		
\usepackage{amsfonts}		
\usepackage{amsthm}		
\usepackage[foot]{amsaddr}		

\usepackage{mathtools}		

\mathtoolsset{%
}

\usepackage[utf8]{inputenc}		
\usepackage[T1]{fontenc}		

\usepackage[
cal=cm,
]
{mathalfa}

\usepackage{dsfont}		

\usepackage{inconsolata}

\usepackage[proportional,tabular,lining,sf,mono=false]{libertine}
\usepackage{acronym}		
\newcommand{\acli}[1]{\emph{\acl{#1}}}		
\newcommand{\acdef}[1]{\emph{\acl{#1}} \textup{(\acs{#1})}\acused{#1}}		

\usepackage[labelfont={bf,small},labelsep=colon,font=small]{caption}	
\usepackage[dvipsnames,svgnames]{xcolor}		
\colorlet{MyRed}{Crimson!75!Black}
\colorlet{MyGreen}{DarkGreen!80!Black}
\colorlet{MyBlue}{MediumBlue}

\colorlet{PrimalColor}{DodgerBlue}
\colorlet{DualColor}{Crimson}
\colorlet{ReducedColor}{Gold}

\newcommand{\para}[1]{\medskip\paragraph{\bfseries #1}}

\usepackage{latexsym}		
\usepackage{wasysym}		

\usepackage{subcaption}		
\usepackage{tikz}		
\usepackage{tikz-3dplot}
\usetikzlibrary{calc,patterns,shapes,intersections}

\usepackage{pgfplots}
\pgfplotsset{compat=1.6}

\usepackage{array}		
\usepackage{booktabs}		
\usepackage[inline,shortlabels]{enumitem}		

\usepackage[kerning=true]{microtype}		

\usepackage{xspace}		

\usepackage[numbers,sort&compress]{natbib}		

\newcommand{\citef}[2][]{\citeauthor{#2} \cite[#1]{#2}}

\usepackage{hyperref}
\hypersetup{
colorlinks=true,
linktocpage=true,
pdfstartview=FitH,
breaklinks=true,
pdfpagemode=UseNone,
pageanchor=true,
pdfpagemode=UseOutlines,
plainpages=false,
bookmarksnumbered,
bookmarksopen=false,
bookmarksopenlevel=1,
hypertexnames=true,
pdfhighlight=/O,
urlcolor=MyRed,linkcolor=MyBlue,citecolor=MyGreen,	
pdftitle={},
pdfauthor={},
pdfsubject={},
pdfkeywords={},
pdfcreator={pdfLaTeX},
pdfproducer={LaTeX with hyperref}
}

\usepackage[sort&compress,capitalize,nameinlink]{cleveref}		
\crefname{assumption}{Assumption}{Assumptions}

\usepackage{algorithm}		
\usepackage{algpseudocode}		

\usepackage{thmtools}		
\usepackage{thm-restate}		

\theoremstyle{plain}
\newtheorem{theorem}{Theorem}		
\newtheorem{corollary}{Corollary}		
\newtheorem{lemma}{Lemma}		
\newtheorem{proposition}{Proposition}		

\newtheorem*{corollary*}{Corollary}		

\theoremstyle{definition}
\newtheorem{definition}{Definition}		
\newtheorem{example}{Example}		

\newtheorem*{definition*}{Definition}		
\newtheorem*{assumption*}{Assumptions}		
\newtheorem*{example*}{Example}		

\theoremstyle{remark}

\newtheorem*{remark*}{Remark}		

\newcounter{proofpart}

\numberwithin{remark}{section}		
\numberwithin{example}{section}		

\usepackage[textwidth=30mm]{todonotes}		

\newcommand{\debug}[1]{#1}		
\newcommand{\revise}[1]{#1}		

\newcommand{\newmacro}[2]{\newcommand{#1}{\debug{#2}}}		
\newcommand{\newop}[2]{\DeclareMathOperator{#1}{\debug{#2}}}		

\DeclarePairedDelimiter{\braces}{\{}{\}}		
\DeclarePairedDelimiter{\bracks}{[}{]}		

\DeclarePairedDelimiter{\abs}{\lvert}{\rvert}		

\DeclarePairedDelimiterX{\inner}[2]{\langle}{\rangle}{#1,#2}		
\DeclarePairedDelimiter{\norm}{\lVert}{\rVert}		
\DeclarePairedDelimiterXPP{\dnorm}[1]{}{\lVert}{\rVert}{_{\ast}}{#1}		

\DeclarePairedDelimiterX{\braket}[2]{\langle}{\rangle}{#1,#2}		

\DeclarePairedDelimiterX{\setdef}[2]{\{}{\}}{#1:#2}		
\DeclarePairedDelimiterXPP{\exclude}[1]{\mathopen{}\setminus}{\{}{\}}{}{#1}

\newcommand{\alt}[1]{#1'}		

\newcommand{\R}{\mathbb{R}}		

\DeclareMathOperator*{\argmax}{arg\,max}		
\DeclareMathOperator*{\argmin}{arg\,min}		
\DeclareMathOperator*{\intersect}{\bigcap}		

\DeclareMathOperator{\bd}{bd}		
\DeclareMathOperator{\bigoh}{\mathcal{O}}		
\DeclareMathOperator{\dist}{dist}		
\DeclareMathOperator{\Hess}{Hess}		
\DeclareMathOperator{\im}{im}		
\DeclareMathOperator{\one}{\mathds{1}}		
\DeclareMathOperator{\rank}{rank}		
\DeclareMathOperator{\relint}{ri}		
\DeclareMathOperator{\supp}{supp}		
\DeclareMathOperator{\vol}{vol}		

\newmacro{\coef}{\lambda}		
\newmacro{\dd}{\:d}		
\newcommand{\subs}{\leftarrow}      

\newcommand{\eps}{\varepsilon}		
\newcommand{\pd}{\partial}		

\newcommand{\insum}{\sum\nolimits}		

\newmacro{\pexp}{p}		
\newmacro{\qexp}{q}		
\newmacro{\rexp}{r}		

\newcommand{\cf}{cf.\xspace}		
\newcommand{\eg}{e.g.,\xspace}		
\newcommand{\ie}{i.e.,\xspace}		
\newcommand{\vs}{vs.\xspace}		

\newcommand{\textpar}[1]{\textup(#1\textup)}		

\newcommand{\txs}{\textstyle}		

\newcommand{\from}{\colon}		

\newcommand{\defeq}{\coloneqq}		

\newmacro{\set}{\mathcal{S}}		

\newmacro{\points}{\mathcal{X}}		
\newmacro{\intpoints}{\points^{\circ}}		
\newmacro{\point}{x}		
\newmacro{\pointalt}{\alt\point}		

\newmacro{\zpoints}{\mathcal{Z}}		
\newmacro{\zpoint}{z}		
\newmacro{\zpointalt}{\alt\zpoint}		

\newmacro{\dpoints}{\mathcal{Y}}		
\newmacro{\dpoint}{y}		
\newmacro{\dpointalt}{\alt\dpoint}		

\newmacro{\base}{p}		
\newmacro{\basealt}{q}		

\newmacro{\test}{\base}		

\newmacro{\open}{\mathcal{U}}		
\newmacro{\closed}{\mathcal{C}}		
\newmacro{\cpt}{\mathcal{K}}		
\newmacro{\nhd}{U}		
\newmacro{\nhdalt}{\alt\nhd}		
\newmacro{\region}{\mathcal{R}}

\newmacro{\run}{t}		
\newmacro{\runalt}{s}		

\newmacro{\runs}{\mathcal{T}}		
\newmacro{\start}{0}		
\newmacro{\horizon}{T}		

\newmacro{\vecspace}{\R^{\vdim}}		
\newmacro{\vdim}{n}		
\newmacro{\vvec}{x}		
\newmacro{\bvec}{e}		
\newmacro{\unitvec}{u}		

\newmacro{\subspace}{\mathcal{W}}		
\newmacro{\wvec}{w}		

\newmacro{\tanhull}{\mathcal{Z}}		
\newmacro{\tanvec}{z}		

\newcommand{\dual}[1]{#1^{\ast}}		
\newmacro{\dspace}{\dual\vecspace}		
\newmacro{\dvec}{v}		
\newmacro{\dbvec}{\eps}		

\newmacro{\ones}{\mathbf{1}}		
\newmacro{\mat}{M}		
\newmacro{\eye}{I}		

\newop{\tspace}{T}		
\newop{\tcone}{TC}		
\newop{\dcone}{\dual\tcone}		
\newop{\ncone}{NC}		
\newop{\pcone}{PC}		

\newmacro{\cvx}{\mathcal{C}}		
\newmacro{\subd}{\partial}		
\newmacro{\hmat}{H}		

\newop{\Opt}{Opt}		
\newop{\Sol}{Sol}		

\newmacro{\obj}{f}		
\newmacro{\objalt}{g}		
\newmacro{\sobj}{F}		

\newmacro{\param}{\theta}		
\newmacro{\params}{\Theta}		

\newmacro{\gvec}{g}		
\newmacro{\vecfield}{V}		

\newmacro{\gbound}{G}		
\newmacro{\vbound}{V}		

\newcommand{\sol}[1][\point]{#1^{\ast}}		
\newcommand{\sols}{\sol[\points]}		

\newmacro{\strong}{\ell}		
\newmacro{\smooth}{L}		
\newmacro{\lips}{G}		

\newmacro{\minmax}{\Phi}		

\newmacro{\minvar}{\point_{1}}		
\newmacro{\minvaralt}{\alt\minvar}		
\newmacro{\minvars}{\points_{1}}		

\newmacro{\maxvar}{\point_{2}}		
\newmacro{\maxvars}{\points_{2}}		
\newmacro{\maxvaralt}{\alt\maxvar}		

\newop{\NE}{NE}		
\newop{\CE}{CE}		
\newop{\CCE}{CCE}		

\newop{\brep}{br}		
\newop{\reg}{Reg}		
\newop{\preg}{\overline{Reg}}		
\newop{\val}{val}		

\newcommand{\strat}{\point}		
\newcommand{\strats}{\points}		
\newcommand{\intstrats}{\relint(\strats)}		
\newcommand{\intstratsPr}{\relint(\strats')}	

\newmacro{\play}{i}		
\newmacro{\playalt}{j}		
\newmacro{\nPlayers}{N}		
\newmacro{\players}{\mathcal{\nPlayers}}		

\newmacro{\pure}{\alpha}		
\newmacro{\purealt}{\beta}		
\newmacro{\purebench}{\hat\pure}		
\newmacro{\nPures}{n}		
\newmacro{\pures}{\mathcal{A}}		
\newmacro{\puresbench}{\hat\pures}		
\newmacro{\supported}{\pures^{\ast}}		

\newcommand{\eq}{\sol}		
\newcommand{\eqs}{\sols}		
\newcommand{\peq}{\sol[\pure]}		

\newmacro{\cost}{c}		
\newmacro{\loss}{\ell}		
\newmacro{\pay}{u}		
\newmacro{\payv}{v}		
\newmacro{\pot}{\obj}		

\newmacro{\game}{\mathcal{G}}		
\newmacro{\gamefull}{\game(\players,\pures,\pay)}		

\newmacro{\fingame}{\Gamma}		
\newmacro{\fingamefull}{\Gamma(\players,\pures,\pay)}		

\newmacro{\mixgame}{\simplex(\fingame)}		

\newop{\Eucl}{\Pi}		
\newop{\logit}{\Lambda}		

\newmacro{\hreg}{h}		
\newmacro{\breg}{D}		
\newmacro{\pmap}{P}		
\newmacro{\mirror}{Q}		
\newmacro{\fench}{F}		
\newmacro{\hstr}{K}		
\newmacro{\depth}{H}		
\newmacro{\zone}{\mathbb{D}}		
\newmacro{\subpoints}{\points^{\circ}}		

\newmacro{\score}{\dpoint}		
\newmacro{\scores}{\mathcal{Y}}		

\DeclareMathOperator{\ex}{\mathbb{E}}		
\DeclareMathOperator{\prob}{\mathbb{P}}		
\DeclareMathOperator{\simplex}{\Delta}		

\newmacro{\sample}{\omega}		
\newmacro{\samples}{\Omega}		
\newmacro{\filter}{\mathcal{F}}		
\newmacro{\probspace}{(\samples,\filter,\prob)}		

\newmacro{\event}{E}       
\newmacro{\eventalt}{H}       

\newmacro{\mean}{\mu}		
\newmacro{\sdev}{\sigma}		
\newmacro{\variance}{\sdev^{2}}		

\newmacro{\dkl}{D_{\mathrm{KL}}}		

\DeclarePairedDelimiterXPP{\exof}[1]{\ex}{[}{]}{}{
 #1}

\DeclarePairedDelimiterXPP{\probof}[1]{\prob}{(}{)}{}{
 #1}

\DeclarePairedDelimiterXPP{\oneof}[1]{\one}{\{}{\}}{}{
 #1}

\newcommand{\orbit}[2][]{\point_{#1}(#2)}		
\newcommand{\dorbit}[2][]{\dpoint_{#1}(#2)}		

\newcommand{\dotdorbit}[2][]{\dot\dpoint_{#1}(#2)}		

\newmacro{\flowmap}{\Phi}		
\newcommand{\flow}[2]{\flowmap_{#1}(#2)}		

\newmacro{\graph}{\mathcal{G}}
\newmacro{\vertices}{\mathcal{V}}
\newmacro{\edges}{\mathcal{E}}

\newmacro{\gmat}{g}		
\newmacro{\gdist}{\dist_{\gmat}}
\newmacro{\ball}{\mathbb{B}}		
\newmacro{\sphere}{\mathbb{S}}		

\usepackage{multirow,array}
\usepackage{comment}
\usepackage[skins]{tcolorbox}
\usepackage{lipsum}

\begin{document}


\newcommand{\longtitle}{\uppercase{No-regret learning and mixed Nash equilibria:\\They do not mix}}		
\newcommand{\runtitle}{\uppercase{No-Regret Learning and Mixed Nash Equilibria: They do not mix}}		

\title[\runtitle]{\longtitle}		

\author[Flokas]{Lampros Flokas$^{\ast}$}
\author[Vlatakis-Gkaragkounis]{Emmanouil V. Vlatakis-Gkaragkounis$^{\ast}$}
\author[Lianeas]{Thanasis Lianeas$^{\dag}$}
\author[Mertikopoulos]{Panayotis Mertikopoulos$^{\ast\ast,\sharp}$}		
\author[Piliouras]{Georgios Piliouras$\ddag$}

\address{$^{\ast}$Department of Computer Science, Columbia University, New York, NY10025.}
\address{$^{\dag}$School of Electrical and Computer Engineering, National Technical University of Athens, Athens, Greece.}
\address{$^{\ast\ast}$Univ. Grenoble Alpes, CNRS, Inria, LIG, 38000, Grenoble, France.}		
\address{$^{\sharp}$Criteo AI Lab.}		
\address{$\ddag$Singapore University of Technology and Design, Singapore.}

\email{lamflokas@cs.columbia.edu}
\email{emvlatakis@cs.columbia.edu}
\email{lianeas@corelab.ntua.gr}
\email{panayotis.mertikopoulos@imag.fr}		
\email{georgios@sutd.edu.sg}
\subjclass[2020]{Primary 91A26, 37N40; Secondary 91A68, 68Q32, 68T05.}
\keywords{Regret; follow the regularized leader; game theory; stability of equilibria.}

\thanks{Acknowledgments of financial support are given in p.~\pageref{ref:thanks}.}

\newacro{LHS}{left-hand side}
\newacro{RHS}{right-hand side}
\newacro{iid}[i.i.d.]{independent and identically distributed}
\newacro{lsc}[l.s.c.]{lower semi-continuous}

\newacro{KKT}{Karush\textendash Kuhn\textendash Tucker}
\newacro{ERD}{Euclidean regularization dynamics}
\newacro{FTRL}{follow the regularized leader}
\newacro{OGD}{online gradient descent}
\newacro{OMD}{online mirror descent}
\newacro{PD}{projection dynamics}
\newacro{MWU}{multiplicative weights update}
\newacro{EW}{exponential weights}
\newacro{RD}{replicator dynamics}
\newacro{NE}{Nash equilibrium}
\newacroplural{NE}[NE]{Nash equilibria}
\newacro{VI}{variational inequality}
\newacroplural{VI}[VIs]{variational inequalities}

\begin{abstract}
%
%

Understanding the behavior of no-regret dynamics in general $\nPlayers$-player games is a fundamental question in online learning and game theory.
A folk result in the field states that, in finite games, the empirical frequency of play under no-regret learning converges to the game's set of coarse correlated equilibria.
By contrast, our understanding of how the day-to-day behavior of the dynamics correlates to the game's \emph{Nash} equilibria is much more limited, and only \emph{partial} results are known for \emph{certain} classes of games (such as zero-sum or congestion games).
In this paper, we study the dynamics of \acdef{FTRL}, arguably the most well-studied class of no-regret dynamics, and we establish a sweeping negative result showing that \emph{the notion of mixed \acl{NE} is antithetical to no-regret learning}.
Specifically, we show that any \acl{NE} which is not \emph{strict} (in that every player has a unique best response) cannot be stable and attracting under the dynamics of \ac{FTRL}.
This result has significant implications for predicting the outcome of a learning process as it shows unequivocally that only strict (and hence, \emph{pure}) \aclp{NE} can emerge as stable limit points thereof.
\end{abstract}

\allowdisplaybreaks		
\acresetall		
\maketitle

\section{Introduction}
\label{sec:introduction}

Regret minimization is one of the most fundamental requirements for online learning and decision-making in the presence of uncertainty and unpredictability \citep{CBL06}.
Defined as the difference between the cumulative performance of an adaptive policy and that of the best fixed action in hindsight, the regret of an agent provides a concise and meaningful benchmark for quantifying the ability of an online algorithm to adapt to an otherwise unknown and unpredictable environment.

Arguably, the most widely studied class of no-regret algorithms is the general algorithmic scheme known as \acdef{FTRL} \citep{SSS06,SS11}.
This umbrella learning framework includes as special cases
the \ac{MWU} \citep{Vov90,LW94,ACBFS95,AHK12} and \ac{OGD} algorithms \citep{Zin03},
both of which achieve a min-max optimal $\bigoh(\horizon^{1/2})$ regret guarantee.
For obvious reasons, the ability of \ac{FTRL} to adapt optimally to an unpredictable environment makes them ideal for applying them in multi-agent environments \textendash\ \ie games.
In this case, if all agents adhere to a no-regret learning process based on \ac{FTRL} (or one of its variants), as the sequence of play becomes more predictable, stronger regret guarantees are achievable, possibly down to constant regret, see \eg
\cite{RS13-COLT,FLST16,SyrgkanisALS15,KM17,MPP18,bailey2019fast,bailey2019finite,MZ19} and references therein.
As such, several crucial questions arise:
\renewenvironment{quote}{%
\list{}{%
	\leftmargin0.8cm   
	\rightmargin\leftmargin
}
\item\relax
}
{\endlist}
\begin{quote}
\centering
What are the \emph{game-theoretic implications} of the no-regret guarantees of \ac{FTRL}?\\
\emph{Do the dynamics of \ac{FTRL} converge to an equilibrium of the underlying game?}
\end{quote}

A folk answer to this question is that ``\emph{no-regret learning converges to equilibrium in all games}'' \citep{NRTV07}, suggesting in this way that no-regret dynamics inherently gravitate towards game-theoretically meaningful states.
However,
at this level of abstraction,
both the \emph{type of convergence} as well as the specific \emph{notion of equilibrium}
that go in this statement
are not as strong as one would have hoped for.
Formally,
the only precise conclusion that can be drawn is as follows:
\emph{under a no-regret learning procedure, the empirical frequency of play converges to the game's set of coarse correlated equilibria} \cite{HMC00,Han57}.

This leads to an important disconnect with standard game-theoretic solution concepts on several grounds.
First, even in $2$-player games, coarse correlated equilibria may be exclusively supported on \emph{strictly} dominated strategies \citep{VZ13}, so they fail even the most basic requirements of rationalizability \citep{DF90,FT91}.
Second, the archetypal game-theoretic solution concept is that of \acdef{NE}, and convergence to a \acl{NE} is a much more tenuous affair:
since no-regret dynamics are, by construction, uncoupled (in the sense that a player's update rule does not \emph{explicitly} depend on the payoffs of other players), the impossibility result of \citef{HMC03} precludes the convergence of no-regret learning to \acl{NE} in \emph{all} games.
This is consistent with the numerous negative complexity results for finding a \acl{NE} \cite{DGP06,rubinstein2018inapproximability}:
an incremental method like \ac{FTRL} simply cannot have enough power to overcome PPAD completeness and converge to \acl{NE} given adversarially chosen initial conditions.

In view of the above, a natural test of whether the dynamics of \ac{FTRL} favor convergence to a \acl{NE} is to see whether they eventually stabilize and converge to it when initialized nearby.
In more precise language,
\emph{are \aclp{NE} asymptotically stable in the dynamics of \ac{FTRL}?}
And, perhaps more importantly, \emph{are all \aclp{NE} created equal in this regard?}


\para{Our contributions}

We establish a stark and robust dichotomy between how the dynamics of \ac{FTRL} treat \aclp{NE} in \emph{mixed} (\ie randomized) \vs \emph{pure} strategies.
For the case of mixed \aclp{NE} we establish a sweeping negative result to the effect that
\emph{the notion of mixed \acl{NE} is antithetical to no-regret learning}.
More precisely, we show that
any \acl{NE} which is not \emph{strict} (in the sense that every player has a unique best response)
cannot be stable and attracting under the dynamics of \ac{FTRL}.
Schematically:

\begin{center}
\textbf{Informal Theorem:}
Asymptotically stable point for \ac{FTRL}
	$\implies$
	Pure \acl{NE}

\textbf{Equivalently:}
Mixed \acl{NE}
	$\implies$
	Not asymptotically stable under \ac{FTRL}
\end{center}

The linchpin of our analysis is the following striking property of the \ac{FTRL} dynamics:
when viewed in the space of ``payoffs'' (their natural state space),
\emph{they preserve volume irrespective of the underlying game.}
More precisely, the Lebesgue measure of any open set of initial conditions in the space of payoffs remains invariant as it is carried along the flow of the \ac{FTRL} dynamics (\cf\cref{fig:grid}).
Importantly, this result is \emph{not} true in the problem's ``primal'' space, \ie the space of the player's mixed strategies:
here, sets of initial conditions can expand or contract indefinitely under the standard Euclidean volume form.

This duality between payoffs and strategies is the leitmotif of our approach and has a number of important consequences.
First, exploiting the volume-preservation property of \ac{FTRL}, we show that no interior \acl{NE}
(and, furthermore, no closed set in the interior of the strategy space)
can be asymptotically stable under the dynamics of \ac{FTRL}, as this effectively would necessitate volume contraction in the interior of the space (\cref{thm:unstable-int}).

To move beyond this result and disqualify \emph{all} non-strict \aclp{NE} (not just interior ones) more intricate arguments are required.
In this case, a fundamental distinction arises between classes of dynamics that may attain the boundary of the players' strategy space in finite time versus those that do not.
The first case concerns \ac{FTRL} dynamics with an everywhere-differentiable regularizer, like the Euclidean regularizer that gives rise to \ac{OGD} and the associated projection dynamics.
The second concerns dynamics where the regularizer becomes \emph{steep} at the boundary of the strategy simplex, \eg like the Shannon-Gibbs entropy that gives rise to the \acf{MWU} algorithm and the 
\acl{RD}.
While the interior of the strategy simplex is invariant for the second class of dynamics, this is not the case for the former:
in Euclidean-like cases, the support of the mixed strategy of an agent may change over time.
This leads to an essential dichotomy in the boundary behavior of different classes of \ac{FTRL} dynamics.
Nonetheless, despite the qualitatively distinct long-run behavior of the dynamics, a unified message emerges:
\emph{under the dynamics of \ac{FTRL}, only strict \aclp{NE} survive} (\cref{thm:unstable-mixed}).

Finally, for the case of steep, entropy-like regularizers we prove that not only their asymptotically stable points but much more generally \emph{any} asypmptotically stable set must contain at least one pure strategy profile (\cref{thm:unstable-set}).

\para{Related work}

The regret properties of \ac{FTRL} have given rise to a vast corpus of literature which we cannot hope to review here;
for an appetizer, we refer the reader to \cite{SS11,BCB12} and references therein.
On the other hand, the long-run behavior of \ac{FTRL} in games (even finite ones) is nowhere near as well understood.
A notable exception to this is the case of the replicator dynamics
which have been studied extensively due to their origins and connection with evolutionary game theory, \cf \cite{TJ78,Wei95,HS98,San10} for a review.
For the replicator dynamics, a special instance of the volume preservation principle was first discovered by \citef{Aki80} and ultimately gave rise to the so-called ``folk theorem'' of evolutionary game theory:%
\footnote{Interestingly, Akin's result was established under a special \emph{non-Euclidean} volume form on the game's \emph{strategy} space, a fact which made any attempts at generalization particularly elusive.}
in population games, the notions of strict \acl{NE} and asymptotic stability coincide \citep{HS03}.
This instability of mixed \aclp{NE} plays a major role in the theory of population games as it shows that even the weakest form of mixing cannot be stable in an evolutionary sense.
The volume preservation result that we establish here can be seen as a much more general ``learning analogue'' of this biological principle and provides an important link between population dynamics and the theory of online learning in games.

Recent work has examined the non-convergence of \ac{FTRL} dynamics in more specialized settings.
 \citef{CGM15} established a version of the folk theorem of evolutionary game theory for a subclass of ``decomposable'', steep \ac{FTRL} dynamics.
By contrast, \citef{MPP18} focused on two-player \emph{zero-sum} games (and networked versions thereof), and showed that almost all trajectories of \ac{FTRL} orbit interior equilibria at a fixed distance without ever converging to equilibrium, generalizing the previous analysis for replicator dynamics by \citef{piliouras2014optimization}.
This is an interior equilibrium avoidance result, but one that \emph{uniquely concerns zero-sum games}. Although the above results apply for continuous-time dynamics, in discrete-time non-convergence results only become stronger.
\citef{bailey2018multiplicative} proved that discrete-time \ac{FTRL} diverges away from the Nash equilibrium in zero-sum games, whereas~\citef{cheung2019vortices} established   
 Lyapunov chaos (volume-expansion, butterfly effects). Understanding the detailed  geometry of non-equilibrating \ac{FTRL} dynamics, e.g., periodicity/chaos, is an interesting direction where volume analysis has found application~\cite{piliouras2017learning,nagarajan2018three,boone2019darwin,nagarajanchaos,bailey2019finite,cheung2020chaos}.
Non-convergence, recurrence results have recently been established for \ac{FTRL} dynamics via volume analysis even outside normal form games, e.g., in non-convex non-concave min-max differential games~\cite{vlatakis2019poincare} and imperfect information zero-sum games~\cite{perolat2020poincar}. Finally, such instability, non-convergence results have inspired new, dynamics-based, solution concepts for games that generalize strict Nash while allowing cyclic, recurrent behavior~\cite{papadimitriou2018nash,omidshafiei2019alpha,rowland2019multiagent,papadimitriou2019game,HMC20}.

In the converse direction, a complementary research thread has shown strict \aclp{NE} are asymptotically stable under several incarnations of the \ac{FTRL} dynamics \cite{CGM15,MS16,MV16,MerSan18,BM17,CHM17-SAGT,MZ19}.
Our paper establishes the \emph{converse} to this stability result, thus leading to the
the following overarching principle (which covers all generic $\nPlayers$-player games):
\begin{center}
\itshape
Asymptotic stability under \ac{FTRL}
	$\iff$
	Strict \acl{NE}
\end{center}
This result has significant implications for predicting the outcome of a learning process as it shows unequivocally that its pointwise stable outcomes are \emph{precisely} the strict (and hence, \emph{pure}) \aclp{NE} of the underlying game.

\section{Preliminaries}
\label{sec:preliminaries}

\para{Notation}
If $\obj$ is a function of a single variable, we will abuse notation slightly and extend it to vector variables $\point\in\R^{\vdim}$ by letting $\obj(\point) \subs (\obj(\point_{1}),\dotsc,\obj(\point_{\vdim}))$.
We will also understand inequalities involving vectors component-wise, \ie $(\point_{1},\dotsc,\point_{\vdim})>0$ means that $\point_{i}>0$ for all $i=1,\dotsc,\vdim$.

\para{The game}
Throughout the sequel, we will focus on finite games.
Formally, a \emph{finite game in normal form} is defined as a tuple $\fingame \equiv \fingamefull$ consisting of
\begin{enumerate*}
[(\itshape i\hspace*{.5pt}\upshape)]
\item
a finite set of \emph{players} $\play \in \players = \{1,\dotsc,\nPlayers\}$;
\item
a finite set of \emph{actions} (or \emph{pure strategies}) $\pures_{\play} = \{\pure_{1},\dotsc,\pure_{\nPures_{\play}}\}$ per player $\play\in\players$;
and
\item
each player's payoff function $\pay_{\play} \from \pures \to \R$, where $\pures \defeq \prod_{\play} \pures_{\play}$ denotes the ensemble of all possible \emph{action profiles} $\pure = (\pure_{1},\dotsc,\pure_{\nPlayers})$.
\end{enumerate*}
In this general context, players can also play \emph{mixed strategies}, \ie probability distributions $\strat_{\play} = (\strat_{\play\pure_{\play}})_{\pure_{\play}\in\pures_{\play}} \in \simplex(\pures_{\play})$ over their pure strategies $\pure_{\play} \in \pures_{\play}$.
Collectively, we will write $\strats_{\play} \defeq \simplex(\pures_{\play})$ for the mixed strategy space of player $\play$ and $\strats \defeq \prod_{\play}\strats_{\play}$ for the space of all mixed strategy profiles $\strat = (\strat_{1},\dotsc,\strat_{\nPlayers})$.

Given a mixed profile $\strat\in\strats$,
the corresponding expected payoff of player $\play$ will be
\begin{equation}
\label{eq:pay}
\pay_{\play}(\strat)
	= \insum_{\pure_{1}\in\pures_{1}} \dotsi \insum_{\pure_{\nPlayers} \in \pures_{\nPlayers}}
		\strat_{1,\pure_{1}} \dotsm \strat_{\nPlayers,\pure_{\nPlayers}}\,
		\pay_{\play}(\pure_{1},\dotsc,\pure_{\nPlayers}).
\end{equation}
To keep track of the payoffs of each individual action, we will also write
\begin{equation}
\label{eq:payv}
\payv_{\play\pure_{\play}}(\strat)
	\defeq \pay_{\play}(\pure_{\play};\strat_{-\play})
\end{equation}
for the payoff of the pure strategy $\pure_{\play}\in\pures_{\play}$ in the mixed profile $\strat = (\strat_{\play};\strat_{-\play}) \in \strats$.%
\footnote{We are using here the standard game-theoretic shorthand $(\strat_{\play};\strat_{-\play}) \defeq (\strat_{1},\dotsc,\strat_{\play},\dotsc,\strat_{\nPlayers})$ to highlight the strategic choice of a given player $\play\in\players$ versus that of the player's opponents $\players_{-\play} \defeq \players\exclude{\play}$.}
Hence, writing $\payv_{\play}(\strat) \defeq (\payv_{\play\pure_{\play}}(\strat))_{\pure_{\play}\in\pures_{\play}} \in \R^{\pures_{\play}}$ for the \emph{payoff vector} of player $\play$, we get the compact expression
\begin{equation}
\label{eq:pay-payv}
\pay_{\play}(\strat)
	= \braket{\payv_{\play}(\strat)}{\strat_{\play}}
	= \insum_{\pure_{\play}\in\pures_{\play}}
		\strat_{\play\pure_{\play}}
		\payv_{\play\pure_{\play}}(\strat)
\end{equation}
where, in standard notation, $\braket{\payv}{\strat} = \payv^{\top}\strat$ denotes the ordinary pairing between $\payv$ and $\point$.

In terms of solutions, the most widely used concept in game theory is that of a \acdef{NE}, \ie a state $\eq\in\strats$ such that
\begin{equation}
\label{eq:Nash}
\tag{NE}
\pay_{\play}(\eq)
	\geq \pay_{\play}(\strat_{\play};\eq_{-\play})
	\quad
	\text{for all $\strat_{\play}\in\strats_{\play}$ and all $\play\in\players$}.
\end{equation}
Writing $\supp(\eq_{\play}) = \setdef{\pure_{\play}\in\pures_{\play}}{\eq_{\play\pure_{\play}} > 0}$ for the support of $\eq_{\play}$,
\aclp{NE} can be equivalently characterized via the \acl{VI}
\begin{equation}
\label{eq:Nash-var}
\payv_{\play\peq_{\play}}(\eq)
	\geq \payv_{\play\pure_{\play}}(\eq)
	\quad
	\text{for all $\peq_{\play}\in\supp(\eq_{\play})$ and all $\pure_{\play} \in \pures_{\play}$, $\play\in\players$}.
\end{equation}
In turn, this characterization leads to the following taxonomy:
\begin{enumerate}
\item
$\eq$ is called \emph{pure} if $\supp(\eq) = \prod_{\play} \supp(\eq_{\play})$ is a singleton.
\item
If $\eq$ is not pure, we say that it is \emph{mixed};
and if $\supp(\eq) = \pures$, we say that it is \emph{fully mixed}.
\end{enumerate}
By definition, pure \aclp{NE} are themselves pure strategies and correspond to vertices of $\strats$;
at the other end of the spectrum, fully mixed equilibria belong to the relative interior $\intstrats$ of $\strats$, so they are often referred to as \emph{interior} equilibria.

Another key distinction between \aclp{NE} concerns the defining inequality \eqref{eq:Nash}:
if this inequality is strict for all $\strat_{\play}\neq\eq_{\play}$, $\play\in\players$, $\eq$ is called itself \emph{strict}.
Strict \aclp{NE} are pure a fortiori, and they play a key role in game theory because any unilateral deviation incurs a strict loss to the deviating player;
put differently, if $\eq$ is strict, \emph{every player has a unique best response}.
Taking this idea further, $\eq$ is called \emph{quasi-strict} if \eqref{eq:Nash-var} is strict for all $\pure_{\play}\in\pures_{\play} \setminus \supp(\eq_{\play})$, \ie if all best responses of player $\play$ are contained in $\supp(\eq_{\play})$.
By a deep result of \citef{Rit94}, all \aclp{NE} are quasi-strict in almost all games;%
\footnote{Specifically, on a set which is open and dense (and hence of full measure) in the space
of all games.}
in view of this, we will tacitly assume in the sequel that all equilibria considered are quasi-strict, a property known as ``genericity'' \citep{FT91,LRS19,CHM17-NIPS}.

\begin{remark*}
We should stress here that quasi-strict equilibria \emph{need not be pure:}
they could be partially or even fully mixed, \eg as in the case of Stag Hunt, Rock-Paper-Scissors, Matching Pennies, the Battle of the Sexes, etc.
We provide a series of illustrative examples in the supplement.
\end{remark*}

\para{Regret}

A key requirement in online learning is the minimization of the players' \emph{regret},
\ie the cumulative payoff difference between a player's mixed strategy at a given time and the player's best possible strategy in hindsight.
In more detail, assuming that play evolves in continuous time $\run\geq0$, the (external) regret of a player $\play\in\players$ relative to a sequence of play $\orbit{\run}\in\strats$ is defined as
\begin{equation}
\reg_{\play}(\horizon)
	= \max_{\test_{\play} \in \strats_{\play}}
		\int_{\start}^{\horizon}
			\bracks{\pay_{\play}(\test_{\play};\orbit[-\play]{\run}) - \pay_{\play}(\orbit{\run})}
		\dd\run,
\end{equation}
and we say that player $\play$ has \emph{no regret} under $\orbit{\run}$ if $\reg_{\play}(\horizon) = o(\horizon)$.

\para{No-regret learning via regularization}
\label{sec:FTRL-intro}

The most widely used method to achieve no-regret is the class of policies known as \acdef{FTRL} \citep{SSS06,SS11}.
Heuristically, at each $\run\geq0$, \ac{FTRL} prescribes a mixed strategy that maximizes the players' cumulative payoff up to time $\run$ minus a regularization penalty which incentivizes exploration.
Formally, this is represented by the dynamics
\begin{align}
\dorbit[\play\pure_{\play}]{\run}
	&= \dorbit[\play\pure_{\play}]{\start}
		+ \int_{\start}^{\run} \payv_{\play\pure_{\play}}(\orbit{\runalt}) \dd\runalt
	\hspace{-4em}
	\tag*{\small\{aggregate payoffs\}}
	\\
\orbit[\play\pure_{\play}]{\run}
	&= \mirror_{\play\pure_{\play}}(\dorbit[\play]{\run})
	\tag*{\small\{choice of strategy\}}
\end{align}
or, in more compact notation:
\begin{equation}
\label{eq:FTRL}
\tag{FTRL}
\dotdorbit{\run}
	= \payv(\mirror(\dorbit{\run})).
\end{equation}
In the above, each $\score_{\play\pure_{\play}}$ plays the role of an auxiliary ``score variable'' which measures the aggregate performance of the pure strategy $\pure_{\play}\in\pures_{\play}$ over time.
These scores are subsequently tranformed to mixed strategies by means of a player-specific \emph{choice map} $\score_{\play} \mapsto \strat_{\play} = \mirror_{\play}(\score_{\play})$ which is defined as
\begin{equation}
\label{eq:choice}
\mirror_{\play}(\score_{\play})
	= \argmax_{\strat_{\play}\in\strats_{\play}} \{\braket{\score_{\play}}{\strat_{\play}} - \hreg_{\play}(\strat_{\play})\}
	\quad
	\text{for all $\score_{\play}\in\scores_{\play} \defeq \R^{\nPures_{\play}}$}.
\end{equation}
In other words, $\mirror_{\play}\from\scores_{\play}\to\strats_{\play}$ essentially acts as a ``soft'' version of the best-response correspondence $\score_{\play} \mapsto \argmax_{\strat_{\play}\in\strats_{\play}} \braket{\score_{\play}}{\strat_{\play}}$, suitably regularized by a convex penalty term $\hreg_{\play}(\strat_{\play})$.
The precise assumptions regarding the \emph{regularizer function} $\hreg_{\play}\from\strats_{\play}\to\R$ will be discussed in detail later;
for now, we provide two prototypical examples of \eqref{eq:FTRL} that will play a major role in the sequel:

\begin{example}
[Entropic regularization and \acl{EW}]
\label{ex:hreg-logit}
One of the most widely used regularizers in online learning is the (negative) Gibbs-Shannon entropy $\hreg_{\play}(\strat_{\play}) = \sum_{\pure_{\play}} \strat_{\play\pure_{\play}} \log\strat_{\play\pure_{\play}}$.
A standard calculation then yields the so-called \emph{logit choice map}, written in vectorized form as
$\logit_{\play}(\score_{\play})
	= 
		{\exp(\score_{\play})}/
		{\sum_{\pure_{\play}\in\pures_{\play}} \exp(\score_{\play\pure_{\play}})}$.
In turn, this leads to the \acli{EW} dynamics:
\begin{equation}
\label{eq:EW}
\tag{EW}
\begin{aligned}
\dotdorbit[\play]{\run}
	&= \payv_{\play}(\orbit{\run}),
	 \\
\orbit[\play]{\run}
	&= 	\logit_{\play}(\dorbit[\play]{\run}).
\end{aligned}
\end{equation}
The system \eqref{eq:EW} describes the mean dynamics of the so-called \acdef{MWU} algorithm (or ``Hedge'');
for an (incomplete) account of its long history, see \cite{Vov90,LW94,KW97,ACBFS95,FS99,CBL06,MM10,AHK12} and references therein.
\end{example}
\smallskip

\begin{example}
[$L^{2}$ regularization]
\label{ex:hreg-Eucl}
Another popular choice of regularizer is the quadratic penalty $\hreg_{\play}(\strat_{\play}) = (1/2) \norm{\strat_{\play}}^{2}$.
In this case, the associated choice map is the Euclidean projector on the simplex, $\Eucl_{\play}(\score_{\play}) = \argmin_{\strat_{\play}\in\strats_{\play}} \norm{\score_{\play} - \strat_{\play}}$, which gives rise to the \acli{ERD}
\begin{equation}
\label{eq:ERD}
\tag{ERD}
\begin{aligned}
\dotdorbit[\play]{\run}
	&= \payv_{\play}(\orbit{\run}),
	\\
\orbit[\play]{\run}
	&= \Eucl_{\play}(\dorbit[\play]{\run}).
\end{aligned}
\end{equation}
\end{example}
\smallskip

Beyond the two prototypical examples discussed above, the origin of the dynamics \eqref{eq:FTRL} can be traced to \citef{SSS06}, \citef{Nes09}, and, via their link to \ac{OMD}, all the way back to \citef{NY83}. Describing the history and literature surrounding these dynamics would take us too far afield, so we do not attempt it.

%

\section{The fundamental dichotomy of \ac{FTRL} dynamics}
\label{sec:FTRL}

To connect the long-run behavior of \eqref{eq:FTRL} to the \aclp{NE} of the underlying game, we must first understand how the players' mixed strategies evolve under \eqref{eq:FTRL}.
Our goal in this section is to provide some background to this question as a precursor to our analysis in \cref{sec:analysis}.
To lighten notation, we will drop in what follows the player index $\play$, writing for example $\strat_{\pure}$ instead of the more cumbersome $\strat_{\play\pure_{\play}}$;
we will only reinstate the index $\play$ if absolutely necessary to avoid confusion.

\subsection{Scores \vs strategies}
To begin, we note that \eqref{eq:FTRL} exhibits a unique duality:
on the one hand, the variables of interest are the players' mixed strategies $\orbit{\run}\in\strats$;
on the other, the dynamics \eqref{eq:FTRL} evolve in the space $\scores$ of the players' score variables $\dorbit{\run}$.
Mixed strategies are determined by the corresponding scores via the players' choice maps $\score \mapsto \strat = \mirror(\score)$, but this is not a two-way street:
as we explain below, the map $\mirror\from\scores\to\strats$ is \emph{not invertible}, so obtaining an autonomous dynamical system on the strategy space $\strats$ is a delicate affair.
In the general case,
invoking standard arguments from convex analysis \citep{Roc70,BC17} we have $\score(t) \in \nabla \hreg(\orbit{\run}) + \pcone(\orbit{\run})$, where
\begin{align}
\label{eq:polar}
\pcone(\strat)
	&= \setdef
		{\score\in\scores}
		{\score_{\pure} \geq \score_{\purealt}
			\text{ for all }
			\pure\in\supp(\strat),
			\purealt \in \pures}
\end{align}
denotes the polar cone to $\strats$ at $\strat$.%
\footnote{In particular, for all $\score\in\pcone(\strat)$, we have $\score_{\pure} = \score_{\purealt}$ whenever $\pure,\purealt\in\supp(\strat)$.
The similarity of this condition to the characterization \eqref{eq:Nash-var} of \aclp{NE} is not a coincidence:
$\eq$ is a \acl{NE} of $\fingame$ if and only if $\payv(\eq) \in \pcone(\eq)$ \citep{CHM17-NIPS,MerSan18}.}

In the entropic case of \cref{ex:hreg-logit}, the logit choice map $\mirror = \logit$ only returns \emph{fully mixed} strategies since $\exp(\score)>0$.
In the relative interior $\intstrats$ of $\strats$, we have by \Cref{eq:polar} that $\pcone(\strat)=\setdef{(t,\dotsc,t)}{t\in\R}$.
As a result, $\logit$ is not surjective;
however, up to a multiple of $(1,\dotsc,1)$, it is \emph{injective}.
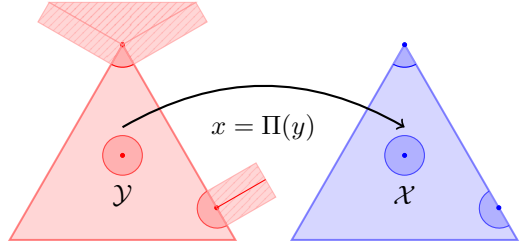
\begin{figure}[t!]
    \centering

\usetikzlibrary{patterns}
\begin{tikzpicture}[scale=0.75]
\filldraw[color = blue!40, opacity = 0.4] (-2,0)  -- (2,0)  -- ({2+4*cos(120)},{4*sin(120})   -- cycle;
\draw[thick,color = blue, opacity = 0.4] (-2,0)  -- (2,0)  -- ({2+4*cos(120)},{4*sin(120})  -- cycle;
\def\yspaceoffset{-5}
\filldraw[color = blue!40, opacity = 0.6] (0,1.5) circle (10pt);
\draw[color = blue, opacity = 0.6] (0,1.5) circle (10pt);
\filldraw[color = blue] (0,1.5) circle (1pt);
\filldraw[color = blue!40, opacity = 0.6] ({cos(120)+2},{sin(120)}) arc (120:300:10pt);
\draw[color = blue, opacity = 0.6] ({cos(120)+2},{sin(120)}) arc (120:300:10pt);
\filldraw[color = blue] ({0.66*cos(120)+2},{0.66*sin(120)}) circle (1pt);
\filldraw[color = blue!40, opacity = 0.6] ({3.6*cos(120)+2},{3.6*sin(120)}) arc (-60:-120:4mm)--({4*cos(120)+2},{4*sin(120)})--cycle ;
\draw[color = blue] ({3.6*cos(120)+2},{3.6*sin(120)}) arc (-60:-120:4mm) ;
\filldraw[color = blue] ({4*cos(120)+2},{4*sin(120)}) circle (1pt);

\draw (0,0.8) node{$\strats$};

\filldraw[color = red!40, opacity = 0.4] (-2+\yspaceoffset,0)  -- (2+\yspaceoffset,0)  -- ({2+4*cos(120)+\yspaceoffset},{4*sin(120})    -- cycle;
\draw[thick,color = red, opacity = 0.4] (-2+\yspaceoffset,0)  -- (2+\yspaceoffset,0)  -- ({2+4*cos(120)+\yspaceoffset},{4*sin(120})    -- cycle;
\filldraw[color = red!40, opacity = 0.6] (0+\yspaceoffset,1.5) circle (10pt);
\draw[color = red, opacity = 0.6] (0+\yspaceoffset,1.5) circle (10pt);
\filldraw[color = red] (0+\yspaceoffset,1.5) circle (1pt);
\filldraw[color = red!40, opacity = 0.6] ({cos(120)+2+\yspaceoffset},{sin(120)}) arc (120:300:10pt);
\draw[color = red, opacity = 0.6] ({cos(120)+2+\yspaceoffset},{sin(120)}) arc (120:300:10pt);
\filldraw[color = red] ({0.66*cos(120)+2+\yspaceoffset},{0.66*sin(120)}) circle (1pt);
\filldraw[color = red!40, opacity = 0.6] ({3.6*cos(120)+2+\yspaceoffset},{3.6*sin(120)}) arc (-60:-120:4mm)--({4*cos(120)+2+\yspaceoffset},{4*sin(120)})--cycle ;
\draw[color = red] ({3.6*cos(120)+2+\yspaceoffset},{3.6*sin(120)}) arc (-60:-120:4mm) ;
\filldraw[color = red] ({4*cos(120)+2+\yspaceoffset},{4*sin(120)}) circle (1pt);
\draw (0+\yspaceoffset,0.8) node{$\scores$};

\filldraw[color = red!40, opacity = 0.4,pattern=north east lines,pattern color=red] 
({cos(120)+2+\yspaceoffset},{sin(120)}) 
-- ({0.3*cos(120)+2+\yspaceoffset},{0.3*sin(120)})
--({0.3*cos(120)+2+sin(120)+\yspaceoffset},{0.3*sin(120)-cos(120)})
-- ({cos(120)+2+sin(120)+\yspaceoffset},{sin(120)-cos(120)})
-- cycle;
\filldraw[color = red!40, opacity = 0.4] 
({cos(120)+2+\yspaceoffset},{sin(120)}) 
-- ({0.3*cos(120)+2+\yspaceoffset},{0.3*sin(120)})
--({0.3*cos(120)+2+sin(120)+\yspaceoffset},{0.3*sin(120)-cos(120)})
-- ({cos(120)+2+sin(120)+\yspaceoffset},{sin(120)-cos(120)})
-- cycle;

\filldraw[color = red] ({4*cos(120)+2+\yspaceoffset},{4*sin(120)}) circle (1pt);

\filldraw[color = red!40, opacity = 0.4,pattern=north east lines,pattern color=red] 
   ({3.6*cos(120)+2+1.5*sin(120)+\yspaceoffset},{3.6*sin(120)-1.5*cos(120)})
-- ({3.6*cos(120)+2+\yspaceoffset},{3.6*sin(120)})
-- ({4*cos(120)+2+\yspaceoffset},{4*sin(120)})
-- ({4*cos(120)+2+1.5*sin(120)+\yspaceoffset},{4*sin(120)-1.5*cos(120)})
-- cycle;
\filldraw[color = red!40, opacity = 0.4] 
   ({3.6*cos(120)+2+1.5*sin(120)+\yspaceoffset},{3.6*sin(120)-1.5*cos(120)})
-- ({3.6*cos(120)+2+\yspaceoffset},{3.6*sin(120)})
-- ({4*cos(120)+2+\yspaceoffset},{4*sin(120)})
-- ({4*cos(120)+2+1.5*sin(120)+\yspaceoffset},{4*sin(120)-1.5*cos(120)})
-- cycle;

\filldraw[color = red!40, opacity = 0.4,pattern=north east lines,pattern color=red] 
   ({3.6*cos(60)-2-1.5*sin(60)+\yspaceoffset},{3.6*sin(60)+1.5*cos(60)}) 
-- ({3.6*cos(60)-2+\yspaceoffset},{3.6*sin(60))})
-- ({4*cos(120)+2+\yspaceoffset},{4*sin(120)})
--  ({4*cos(60)-2-1.5*sin(60)+\yspaceoffset},{4*sin(60)+1.5*cos(60)})
-- cycle;
\filldraw[color = red!40, opacity = 0.4] 
   ({3.6*cos(60)-2-1.5*sin(60)+\yspaceoffset},{3.6*sin(60)+1.5*cos(60)}) 
-- ({3.6*cos(60)-2+\yspaceoffset},{3.6*sin(60))})
-- ({4*cos(120)+2+\yspaceoffset},{4*sin(120)})
--  ({4*cos(60)-2-1.5*sin(60)+\yspaceoffset},{4*sin(60)+1.5*cos(60)})
-- cycle;
\filldraw[color = red!40, opacity = 0.4] 
    ({4*cos(120)+2+1.5*sin(120)+\yspaceoffset},{4*sin(120)-1.5*cos(120)})
--  ({4*cos(120)+2+\yspaceoffset},{4*sin(120)})
--  ({4*cos(60)-2-1.5*sin(60)+\yspaceoffset},{4*sin(60)+1.5*cos(60)})
-- cycle;
\filldraw[color = red!40, opacity = 0.4,pattern=north east lines,pattern color=red] 
    ({4*cos(120)+2+1.5*sin(120)+\yspaceoffset},{4*sin(120)-1.5*cos(120)})
--  ({4*cos(120)+2+\yspaceoffset},{4*sin(120)})
--  ({4*cos(60)-2-1.5*sin(60)+\yspaceoffset},{4*sin(60)+1.5*cos(60)})
-- cycle;
\draw[color = red] ({0.66*cos(120)+2+\yspaceoffset},{0.66*sin(120)})--({0.66*cos(120)+2+sin(120)+\yspaceoffset},{0.66*sin(120)-cos(120)}) ;

\draw [<-,thick]  (0,2) to [in=30,out=150] (0+\yspaceoffset,2) ;
\draw (\yspaceoffset/2,2) node {$\strat = \Eucl(\score)$};
\end{tikzpicture}
    \caption{The inverse images of neighborhoods of different points in $\strats$ under the Euclidean choice map $\mirror = \Eucl$.}
    \label{fig:polarcone}
\end{figure}
On the other hand, in the Euclidean framework of \cref{ex:hreg-Eucl},
the choice map $\mirror = \Eucl$ can also return non-fully mixed strategies. Both \Cref{eq:polar} and \cref{fig:polarcone} show that on the boundary $\pcone(\strat)$ is strictly larger compared to the interior.
Thus $\Eucl$ is surjective but not injective, even modulo a subspace of $\scores$.

The key obstacle to mapping the dynamics \eqref{eq:FTRL} to $\strats$ is the lack of injectivity of $\mirror$.
In turn, this allows us to make two key observations:
\begin{enumerate*}
[(\itshape i\hspace*{.5pt})]
\item
there is an important split in behavior between boundary and interior states;
and
\item
this split is linked to whether the underlying choice map is surjective or not.
\end{enumerate*}
We elaborate on this below.

\subsection{The steep/non-steep dichotomy}

The lack of injectivity of $\logit$ on $\intstrats$ is a technical artifact of the sum-to-one constraints of the strategy probabilities:
knowing all but one of the strategy probabilities 
we can easily recover the remaining one.
Thus the $\scores$ space, having the same number of coordinates as the $\strats$ space,
also contains redundant information.
With an appropriate projection we can remove this redundancy and restore injectivity in the interior, deriving the dynamics of $\orbit{\run}$ on $\strats$.
Making this argument precise for the entropic case of \cref{ex:hreg-logit},
we obtain the \acli{RD}:
\begin{equation}
\label{eq:RD}
\tag{RD}
\dot\strat_{\pure}
	= \strat_{\pure} \bracks*{\payv_{\pure}(\strat) - \pay(\strat)}.
\end{equation}

On the other hand, this is not enough for the Euclidean framework of \cref{ex:hreg-Eucl}.
When trajectories approach $\bd(\strats)$, the positivity constraints $\strat_\play \geq 0$ kick in finite time.
Unlike the sum-to-one constraints of the previous case, these cannot be resolved with a dimensionality reduction so we cannot obtain a well-posed dynamical system on $\strats$ as above.
This problem can only be temporarily avoided for time intervals where $\supp(\orbit{\run})$ remains constant.
For these intervals $\orbit{\run}$ can be shown to satisfy the \acli{PD} \citep{MS16}
\begin{equation}
\label{eq:PD}
\tag{PD}
\dot\strat_{\pure}
	= \payv_{\pure}(\strat) - \abs{\supp(\strat)}^{-1} \insum_{\purealt\in\supp(\strat)} \payv_{\purealt}(\strat) \quad
	\text{if $\pure\in\supp(\strat)$}.
\end{equation}
In contrast to the replicator dynamics, different trajectories of \eqref{eq:PD} can merge or split any number of times, and they may transit from one face of $\strats$ to another in finite time \cite{MS16,MerSan18}.

The two cases above are not just conveniently chosen examples, but archetypes of the fundamentally different behaviors that can be observed under \eqref{eq:FTRL} for different regularizers.
As we discuss in the supplement, this polar split is intimately tied to the behavior of the derivatives of $\hreg$ at the boundary of $\strats$.
To formalize this, we say that $\hreg$ is \emph{steep} if $\norm{\nabla\hreg(\strat)} \to \infty$ whenever $\strat\to\bd(\strats)$;
by contrast, if $\sup_{\strat\in\strats} \norm{\nabla\hreg(\strat)} < \infty$, we say that $\hreg$ is \emph{non-steep}.
Thus, in terms of our examples, the negentropy function of \cref{ex:hreg-logit} is the archetype for steep regularizers, while the $L^{2}$ penalty of \cref{ex:hreg-Eucl} is the non-steep one.
The split between steep and non-steep dynamics may then be stated as follows:
\begin{enumerate}
\item
If $\hreg$ is steep, the mixed-strategy trajectories $\orbit{\run} = \mirror(\dorbit{\run})$ carry all the information required to predict the evolution of the system;
in particular, $\orbit{\start}$ fully determines $\orbit{\run}$ for all $\run\geq0$, and $\orbit{\start}$ remains fully mixed for all time.
\item
If $\hreg$ is non-steep, the trajectories $\orbit{\run} = \mirror(\dorbit{\run})$ do not fully capture the state of the system:
$\orbit{\start}$ \emph{does not} determine $\orbit{\run}$ for all $\run\geq0$, and even the times when $\orbit{\run}$ changes support cannot be anticipated by knowing $\orbit{\start}$ alone.
For concision, we defer the precise statement and proof of this dichotomy to the paper's supplement.
\end{enumerate}

\section{Convergence analysis and results}
\label{sec:analysis}

We now turn to the equilibrium convergence properties of \eqref{eq:FTRL}.
The central question that we seek to address here is the following:
\emph{Which \aclp{NE} can be stable and attracting under \eqref{eq:FTRL}?}
\emph{Are all equilibria created equal in that regard?}
\subsection{Notions of stability}
At a high level, a point is
\begin{enumerate*}
[(\itshape a\upshape)]
\item
\emph{stable} when every trajectory that starts nearby remains nearby;
and
\item
\emph{attracting} when it attracts all trajectories that start close enough.
\end{enumerate*}
Already, this heuristic shows that defining these notions for \eqref{eq:FTRL} is not straightforward:
the target points are strategy profiles in $\strats$, while the dynamics \eqref{eq:FTRL} evolve in the dual space $\scores$.
When $\hreg$ is \emph{steep},
we can define an equivalent presentation of \eqref{eq:FTRL} on $\strats$, so this problem can be circumvented by working solely with mixed strategies;
however, when $\hreg$ is \emph{non-steep}, this is no longer possible and we need to navigate carefully between $\strats$ and $\scores$.
In view of this, we have the following definitions:
\begin{itemize}
\item
$\eq\in\strats$ is \emph{stable} if, for every neighborhood $\nhd$ of $\eq$ in $\strats$, there exists a neighborhood $\alt\nhd$ of $\eq$ such that $\orbit{\run} = \mirror(\dorbit{\run}) \in \nhd$ for all $\run\geq0$ whenever $\orbit{\start} = \mirror(\dorbit{\start}) \in \alt\nhd$.
\item
$\eq\in\strats$ is \emph{attracting} if there exists a neighborhood $\nhd$ of $\eq$ in $\strats$ such that $\orbit{\run} = \mirror(\dorbit{\run}) \to \eq$ whenever $\orbit{\start} = \mirror(\dorbit{\start}) \in \nhd$.
\item
$\eq\in\strats$ is \emph{asymptotically stable} if it is both stable and attracting.
\end{itemize}
For obvious reasons, asymptotic stability is the ``gold standard'' for questions pertaining to equilibrium convergence and it will be our litmus test for the appropriateness of an equilibrium $\eq\in\strats$ as an outcome of play.
Specifically, if a \acl{NE} is not asymptotically stable under \eqref{eq:FTRL}, it is not reasonable to expect a no-regret learner to converge to it, meaning in turn that it cannot be justified as an end-state of the players' learning process.
We expound on this below.

\subsection{Volume preservation}

A key observation regarding asymptotic stability is that neighborhoods of initial conditions near an asymptotically stable point should ``contract'' over time, eventually shrinking down to the point in question.
Our first result below provides an apparent contradiction to this principle:
it shows that volume is preserved under \eqref{eq:FTRL},
\emph{irrespective of the underlying game.}

\begin{restatable}
{proposition}{volumeY}
\label{prop:volume}
Let $\region_{\start} \subseteq \scores$ be a set of initial conditions for \eqref{eq:FTRL} and let $\region_{\run} = \setdef{\dorbit{\run}}{\dorbit{\start} \in \region_{\start}}$ denote its evolution under \eqref{eq:FTRL} after time $\run\geq\start$.
Then, $\vol(\region_{\run}) = \vol(\region_{\start})$.
\end{restatable}

\Cref{prop:volume} (which we prove in the supplement through an application of Liouville's formula) is surprising in its universality as it holds for \emph{all games} and \emph{all instances} of \eqref{eq:FTRL}.
As such, it provides a blanket generalization of the well-known volume-preserving property for the replicator dynamics established by \citet{Aki80},
\revise{as well as subsequent results}
for zero-sum games \citep{MPP18}.


\begin{figure}[t]
\centering

\begin{tikzpicture}[scale=.45]
\small
	
	\def\dualoffset{8.5}
	\def\steep{2*\dualoffset}
	\def\radius{.4}
	\foreach \x in {-3,-2,-1,0,1,2,3}
	\foreach \y in {-3,-2,-1,0,1,2,3}
	{
	\draw[fill=black!10] (\x,\y) circle (1pt);
	\draw[fill=black!10] (\x+\dualoffset,\y) circle (1pt);
	}
	\draw (-3,-3) -- (-3,3) -- (3,3) -- (3,-3) -- (-3,-3);
	\draw (-3+\steep,-3) -- (-3+\steep,3) -- (3+\steep,3) -- (3+\steep,-3) -- (-3+\steep,-3);

	\draw[-stealth,thick] (-3.5+\dualoffset,0) -- (3.5+\dualoffset,0);
	\draw[-stealth,thick] (0+\dualoffset,-3.5) -- (0+\dualoffset,3.5);
	
	\foreach \k in {0,1,2}
	{
	\draw[fill=PrimalColor,opacity=0.3] (\k,\k) circle (\radius);
	\draw[fill=DualColor,opacity=0.3] (\k+\dualoffset,\k) circle (\radius);
	}
	\filldraw[fill=PrimalColor,opacity=0.3] 
	(3-\radius,3) arc[start angle=180, end angle=270, radius=\radius]--(3,3-\radius)--(3,3)--(3-\radius,3);
	\draw[-, ultra thick,PrimalColor] (3,3-\radius)--(3,3)--(3-\radius,3);
	\draw[fill=DualColor,opacity=0.3] (3+\dualoffset,3) circle (\radius);

	\draw[fill=PrimalColor,opacity=0.3] (\steep,0) circle (\radius);
	\draw[fill=PrimalColor,opacity=0.3] (\steep + 1.38635, 1.38635) circle (\radius/1.38635);
	\draw[fill=PrimalColor,opacity=0.3] (\steep + 2.28478, 2.28478) circle (\radius/2.28478);
	\draw[fill=PrimalColor,opacity=0.3] (\steep + 2.71544, 2.71544) circle (\radius/2.71544);

	\foreach \k in {0,1,2}
	{
	\draw[-stealth,thick] (0.3+\k,0.3+\k) -- (0.7+\k,0.7+\k);
	\draw[-stealth,thick] (0.3+\k+\dualoffset,0.3+\k) -- (0.7+\k+\dualoffset,0.7+\k);
	}
	
	\foreach \signx in {-1,1}
	\foreach \signy in {-1,1}
	\foreach \kx in {0,1,2,3}
	\foreach \ky in {0,1,2,3}
	{
	\draw[fill=black!10] ({\steep + 3*\signx*( (exp(\kx)-1)/(exp(\kx)+1) )},{ 3*\signy*(exp(\ky)-1)/(exp(\ky)+1)}) circle (1pt);
	}
	

 
	\draw[-stealth,thick] (0.1+\steep,0.1) -- (1.2+\steep,1.2);
	\draw[-stealth,thick] (1.5+\steep,1.5) -- (2.2+\steep,2.2);
	\draw[-stealth,thick] (2.4+\steep,2.4) -- (2.6+\steep,2.6);
	
	\draw (\dualoffset,3.5) node {\normalsize \eqref{eq:ERD} \;\;\; \eqref{eq:EW}\hphantom{x}};

	\draw (-3,3.5) node {\normalsize\eqref{eq:PD}};
	\draw [stealth-] (1,3.3) to [out=20,in=160] (-2.5+\dualoffset,3.3) ;
	\fill (0,3.75) node {$\strat = \Eucl(\dpoint)$};
	\fill (0,-4) node {Strategy space ($\strats$)};

	\draw (-1+2.5*\dualoffset,3.5) node {\normalsize\eqref{eq:RD}};
	\draw [-stealth] (2+\dualoffset,3.3) to [out=20,in=160] (-1+\dualoffset+\dualoffset,3.3) ;
	\fill (\steep,3.75) node {$\strat = \logit(\dpoint)$};
	\fill (\steep,-4) node {Strategy space ($\strats$)};

	\fill (\dualoffset,-4) node {Score space ($\dpoints$)};
\end{tikzpicture}
\caption{The duality between scores and strategies under \eqref{eq:FTRL}:
the dynamics are volume-preserving in $\scores$, but a volume of initial conditions could either collapse in finite time (in the Euclidean case, left), or shrink asymptotically (in the logit case, right).
This is due to the vastly different geometric properties of each system.}
\label{fig:grid}
\end{figure}
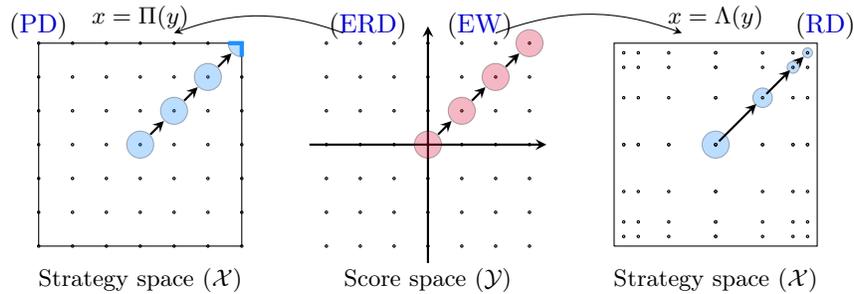


\subsection{Instability of fully mixed equilibria}

As stated above, the volume-preserving property of \eqref{eq:FTRL} would seem to suggest that no strategy can be asymptotically stable.
However, this is a figment of the duality between strategy and score variables:
a mixed strategy orbit $\orbit{\run} = \mirror(\dorbit{\run})$ could converge in $\strats$, even though the corresponding dual orbit $\dorbit{\run}$ diverges in $\scores$ (for an illustration, see \cref{fig:grid} above).
This again brings into sharp contrast the behavior of \eqref{eq:FTRL} at the boundary of $\strats$ versus its behavior at the interior.
Our first instability result below shows that the volume-preserving property of \eqref{eq:FTRL} rules out the stability of \emph{any} fully mixed equilibrium, in \emph{any} game:

\begin{restatable}{theorem}{stableInt}
\label{thm:unstable-int}
A fully mixed \acl{NE} cannot be asymptotically stable under \eqref{eq:FTRL}.
\end{restatable}

The main idea of the proof of \cref{thm:unstable-int} relies on a tandem application of \cref{prop:volume} together with the dimensionality reduction idea we discussed for the entropic case in \cref{sec:FTRL}.
In the resulting quotient space, the inverse image of an interior point $\eq\in\intstrats$ is a single point and the induced dynamics remain volume-preserving.
If $\eq$ is asymptotically stable, a limit point argument rules out the possibility of a trajectory entering and exiting a small neighborhood of its preimage infinitely many times.
At the same time, Lyapunov stability and volume preservation imply that the dynamics are locally recurrent.
This contradicts the transient property established above and proves that $\eq$ cannot be asymptotically stable;
the details involved in making these arguments precise are fairly intricate, so we defer the proof of \cref{thm:unstable-int} to the supplement.

This universal instability result has significant implications as it provides a dynamic justification of the fragility of fully mixed \aclp{NE}.
\Cref{thm:unstable-int} illustrates this principle through the lens of regret minimization:
any deviation from a fully mixed equilibrium invariably creates an opportunity that can be exploited by a no-regret learner.
When every player adheres to such a policy, this creates a vicious cycle which destroys any chance of stability for fully mixed equlibria.

\subsection{The case of partially mixed equilibria}
Taking this premise to its logical extreme, a natural question that arises is whether this instability persists as long as even a \emph{single} player employs a mixed strategy at equilibrium. In the previous case, after the dimensionality reduction argument we described in \cref{sec:FTRL}, \revise{neighborhoods of} fully mixed equilibria \revise{in the space of strategies ($\strats$) correspond to sets of finite volume in the space of payoffs ($\scores$).}
On the contrary, the case of \emph{partially} mixed equilibria is much more complex
\revise{because neighborhoods of points on the boundary of $\strats$ correspond to sets of \emph{infinite} volume in the space of payoffs} 
\textendash\ and this, even after dimensionality reduction (\cf \cref{fig:polarcone}).
Because of this, volume preservation arguments cannot rule out asymptotic stability of \revise{\aclp{NE} lying at the boundary of the strategy space:}
indeed, pure \aclp{NE} also lie on the boundary but they \emph{can} be asymptotically stable \citep{CGM15,MS16,CHM17-NIPS,MerSan18}.

\revise{In view of the above, it is not a priori clear whether partially mixed equilibria would behave more like pure or fully mixed ones \textendash\ or if no conclusion can be drawn whatsoever.
Our next result shows that the dynamics of \ac{FTRL} represent a very sharp selection mechanism in this regard:}


\begin{restatable}{theorem}{stableMixed}
\label{thm:unstable-mixed}
Only strict \aclp{NE} can be asymptotically stable under \eqref{eq:FTRL}.
\end{restatable}

\begin{corollary}
\label{cor:stable-strict}
If $\eq$ is partially mixed, it cannot be asymptotically stable under \eqref{eq:FTRL}.
\end{corollary}

Viewed in isolation, \cref{thm:unstable-int} would seem to be subsumed by \cref{thm:unstable-mixed}, but this is not so:
the former plays an integral role in the proof of the latter, so it cannot be viewed as a special case.
In more detail, the proof of \cref{thm:unstable-mixed} builds on \cref{thm:unstable-int} along two separate axes, depending on whether the underlying regularizer is steep or not:

\begin{enumerate}
[leftmargin=*]
\item
In the steep case, as we discussed in \cref{sec:FTRL} there is a well-posed dynamical system on $\strats$.
As we show in the supplement, each face of $\strats$ is forward-invariant in this system,
so $\eq$ must also be asymptotically stable when constrained to the face $\eqs$ of $\strats$ spanned by $\supp(\eq)$.
The conclusion of \cref{thm:unstable-mixed} then follows by noting that $\eq$ is interior in $\eqs$ and applying \cref{thm:unstable-int} to the restriction of the underlying game to $\eqs$.
\item
The non-steep case is considerably more difficult because \eqref{eq:FTRL} no longer induces a well-posed system on $\strats$.
In lieu of this, by examining the finer structure of the inverse image of $\eq$, it is possible to show the following:
for every small enough compact neighborhood $\cpt$ of $\eq$ in $\strats$, there exists a finite time $\tau_{\cpt} \geq 0$ such that $\supp(\orbit{\run}) = \supp(\eq)$ for all $\run\geq\tau_{\cpt}$ whenever $\orbit{\start} \in \cpt$.
As it turns out, the dynamics after $\run\geq\tau_{\cpt}$ \revise{loca}lly coincide with the mixed strategy dynamics of \eqref{eq:FTRL} applied to the restriction of the underlying game to the face $\eqs$ of $\strats$ spanned by $\eq$.
Since $\eq$ is a fully mixed equilibrium in this restricted game, it cannot be asymptotically stable.
\end{enumerate}

\subsection{Stable limit sets}

We conclude our analysis with a result concerning more general behaviors whereby the dynamics of \ac{FTRL} do not converge to a point, but to a more general \emph{invariant set} \textendash\ such as a chain of stationary points interconnected by solution orbits, a structure known as a \emph{heteroclinic cycle} \citep[see \eg][and references therein]{HS98,San10}. 
As an example, in the case of two-player zero-sum games with a fully mixed equilibrium, it is known that the trajectories of \eqref{eq:FTRL} form periodic orbits (cycles).
However, these orbits are \emph{not} asymptotically stable:
if the initialization of the \ac{FTRL} dynamics is slightly perturbed, the resulting trajectory will be a different periodic orbit, which does not converge to the first (in the language of dynamical systems, the cycles observed in zero-sum games are not \emph{limit cycles}).
We are thus led to the following natural question:
\begin{center}
\emph{What type of invariant structures can arise as stable limits of \eqref{eq:FTRL}?}
\end{center}
To state this question formally, we will require the setwise version of asymptotic stability:
a set $\set$ is called \emph{asymptotically stable} under \eqref{eq:FTRL} if
\begin{enumerate*}
[\itshape a\upshape)]
\item
all orbits $\orbit{\run} = \mirror(\dorbit{\run})$ of \eqref{eq:FTRL} that start sufficiently close to $\set$ remain close;
and
\item
all orbits that start nearby eventually converge to $\set$.
\end{enumerate*}
Then, focusing on the case of steep dynamics to avoid more complicated statements, we have:

\begin{restatable}{theorem}{stableSet}
\label{thm:unstable-set}
Every asymptotically stable set of steep \eqref{eq:FTRL} contains a pure strategy.
\end{restatable}

The proof of \cref{thm:unstable-set} relies on an ``infinite descent'' argument whereby the faces of $\points$ that intersect with $\set$ are eliminated one-by-one, until only pure strategies remain as candidate elements of $\set$ with minimal support;
we provide the details in the supplement.

The importance of \cref{thm:unstable-set}
lies in that it provides a succinct criterion for identifying possible attracting sets of \eqref{eq:FTRL}.
Indeed, by Conley's decomposition theorem (also known as the ``fundamental theorem of dynamical systems'') \citep{Con78}, the flow of \eqref{eq:FTRL} in an arbitrary game decomposes into a chain recurrent part and an attracting part~(see \cite{papadimitriou2018nash,papadimitriou2019game} for several examples/discussion in the case of replicator dynamics).
The recurrent part is exemplified by the periodic orbits that arise in zero-sum games with an interior equilibrium (there are no attractors in this case) \citep{MPP18}.
\Cref{thm:unstable-set} goes a long way to showing that the attracting part of \eqref{eq:FTRL} always intersects the \emph{extremes} of the game's strategy space \textendash\ \ie the players' set of \emph{pure} strategies.
A special case of \Cref{thm:unstable-set}, in the case of replicator dynamics, was employed in~\cite{omidshafiei2019alpha} as a step in the definition of new, dynamics/decomposition-based solution concepts.
Formalizing the exact form of this decomposition in arbitrary games is an open direction for future research with far-reaching implications for the theory of online learning in games.

\section{Concluding remarks}
\label{sec:conclusion}

The well known universal existence theorem for (mixed) Nash equilibria in general games has been very influential not only from a mathematics perspective but also from a public policy one as it seems to suggest that there is no inherent tension in any societal setting between the single-minded pursuit of individual profits and societal stability.
Nash equilibria satisfy both desiderata simultaneously. 
Thus, there is in principle no need for centralized intervention and guidance as market forces will converge upon such a solution.

Our results present an argument in the opposite direction. Unless the game has a pure Nash equilibrium, which is definitely not satisfied in numerous strategic interactions, then societal systems do not self-stabilize, even if they are driven by our most effective payoff seeking dynamics, i.e., gradient learning and its follow-the-regularizer-leader variants. Exploring the tradeoffs between individual optimality and societal stability is thus a much more subtle issue than it first meets the eye, and we hope that we inspire follow-up work that can elucidate these questions further.

\section*{Acknowledgments}
\label{ref:thanks}
\small
%
%
This research was partially supported by the COST Action CA16228 ``European Network for Game Theory'' (GAMENET), the French National Research Agency (ANR) under grant ALIAS, and the Onassis Foundation undr Scholarship ID: F ZN 010-1/2017-2018.

E.V.~Vlatakis-Gkaragkounis is grateful to be supported by NSF grants CCF-1703925, CCF-1763970,
CCF-1814873, CCF-1563155, and by the Simons Collaboration on Algorithms and Geometry. 

T.~Lianeas is supported by the Hellenic Foundation for Research and Innovation (H.F.R.I.) under the ``First Call for H.F.R.I. Research Projects to support Faculty members and Researchers and the procurement of high-cost research equipment grant'',  project BALSAM, HFRI-FM17-1424.

P.~Mertikopoulos is grateful for financial support by
the French National Research Agency (ANR) in the framework of
the ``Investissements d'avenir'' program (ANR-15-IDEX-02),
the LabEx PERSYVAL (ANR-11-LABX-0025-01),
and
MIAI@Grenoble Alpes (ANR-19-P3IA-0003).

G.~Piliouras gratefully acknowledges AcRF Tier-2 grant (Ministry of Education – Singapore) 2016-T2-1-170, grant PIE-SGP-AI-2018-01, NRF2019-NRF-ANR095 ALIAS grant and NRF 2018 Fellowship NRF-NRFF2018-07 (National Research Foundation Singapore).
\normalsize

\appendix
\numberwithin{equation}{section}		
\numberwithin{lemma}{section}		
\numberwithin{proposition}{section}		
\numberwithin{theorem}{section}		

\section{An ontology of \aclp{NE}: representative examples}
\label{app:examples}


In the archetypal game of Prisoner's Dilemma (left), it is easy to check that the unique \acl{NE} is the mutual betrayal which is strict (and hence pure).
On the other hand, Matching Pennies (right) is an example of a zero-sum game whose unique \acl{NE} is fully mixed but still quasi-strict (since all strategies present in its support are unilateral best responses to it).
We mention the above to clarify that quasi-strict \emph{does not mean} pure equilibria and includes also the fully mixed \acl{NE};
the terminology is, perhaps, unfortunate, but otherwise deeply entrenched in the game-theoretic literature \citep{FT91}.


\begin{table}[htbp]
\centering
\setlength{\extrarowheight}{\smallskipamount}

\begin{tabular}{cc}
\begin{tabular}{cc|c|c|}
	& \multicolumn{1}{c}{}	& \multicolumn{2}{c}{Player $Y$}
		\\
	& \multicolumn{1}{c}{}	&\multicolumn{1}{c}{$B$}	&\multicolumn{1}{c}{$S$}
		\\
	\cline{3-4}
	\multirow{2}*{Player $X$}	& $B$ & $(3,3)$ & $(0,5)$ \\\cline{3-4}
	& $S$ & $(5,0)$ & $(1,1)$ \\\cline{3-4}
\end{tabular}&
\begin{tabular}{cc|c|c|}
	& \multicolumn{1}{c}{} & \multicolumn{2}{c}{Player $Y$}\\
	& \multicolumn{1}{c}{} & \multicolumn{1}{c}{$H$}	& \multicolumn{1}{c}{$T$} \\\cline{3-4}
	\multirow{2}*{Player $X$}	& $H$ & $(1,-1)$ & $(-1,1)$ \\\cline{3-4}
	& $T$ & $(-1,1)$ & $(1,-1)$ \\\cline{3-4}
\end{tabular}
\end{tabular}

\medskip
\caption{Prisoner's Dilemma (left) \& Matching Pennies (right).}
\label{tab:games}
\end{table}



\section{Basic properties of the \ac{FTRL} dynamics}
\label{app:structural}

\subsection{Definitions from dynamical systems}

In this appendix, we provide some general preliminaries from general topology and the theory of dynamical systems that we will use freely in the sequel.

 
A key notion in our analysis is that of \emph{\textpar{Poincaré} recurrence}.
Intuitively, a dynamical system is recurrent if, after a sufficiently long (but \emph{finite}) time, almost every state returns arbitrarily close to the system's initial state.%
\footnote{Here, ``almost'' means that the set of such states has full Lebesgue measure.}
More formally, given a dynamical system on $\points$ that is defined by means of a \emph{semiflow}
$\flowmap\from\points\times[0,\infty)\to\points$,
we have:%
\footnote{A smooth map $\flowmap\from\points\times[0,\infty)\to\points$ is called a \emph{semiflow} if $\flow{0}{\point} = \point$ and $\flow{\run+\runalt}{\point} = \flow{\runalt}{\flow{\run}{\point}}$ for all $\point\in\points$ and all $\run,\runalt\geq0$.
Heuristically, $\flow{\run}{\point} \equiv \flow{\run}{\point}$ describes the trajectory of the dynamical system starting at $\point$.}

\begin{definition}
\label{def:recurrence}
A point $\point\in\points$ is said to be \emph{recurrent} under $\flowmap$ if, for every neighborhood $U$ of $\point$ in $\points$, there exists an increasing sequence of times $t_{n}\uparrow\infty$ such that $\flow{\run_{n}}{\point} \in U$ for all $n$.
Moreover, the flow $\flowmap$ is called \emph{\textpar{Poincaré} recurrent} if, for every measurable subset $A$ of $\points$, the set of recurrent points in $A$ has full measure.
\end{definition}
The above definition directly implies that the flow $\flow{\run}{\point}$ from a recurrent point $\point$ cannot converge to any $\pointalt \neq \point$.
Poincaré's recurrence theorem gives sufficient condition for the existence of such points.
 
\begin{theorem}[Poincaré Recurrence Theorem]\label{thm:poincare_reccurrence}
If a flow $\flowmap$ preserves volume and its orbits are bounded,
then almost every point is recurrent under $\flowmap$.
 \end{theorem}

The key notion in the above formulation of the theorem is that of volume preservation:
formally, a flow $\flowmap$ is \emph{volume-preserving} if $\vol(\flow{\run}{\region}) = \vol(\region)$ for any set of initial conditions $\region\subseteq\points$.
A useful condition to establish this property is via \emph{Liouville's formula}, as stated below:

\begin{theorem}
[Liouville's formula]
\label{thm:Liouville}
Let $\Phi$ be the flow of a dynamical system with infinitesimal generator $\vecfield$, \ie $\flow{\run}{\point}$ is the solution trajectory of the ordinary differential equation
\begin{equation}
\frac{d}{d\run}\orbit{\run}
	= \vecfield(\orbit{\run})
\end{equation}
with initial condition $\orbit{\start} = \point$.
Then, letting $\region_{\run} = \flow{\run}{\region}$ for an arbitrary measurable set $\region$, we have
\begin{equation}
\label{eq:Liouville}
\frac{d}{dt} \vol[\region_{\run}]
	= \int_{\region_{\run}} \operatorname{div}[\vecfield(\point)] \dd\point
\end{equation}
\end{theorem}

\begin{corollary}
If $\vecfield$ is \emph{incompressible} over $\vecspace$ \textpar{\ie $\operatorname{div}\vecfield(\point) = 0$ for all $\point\in\vecspace$}, the induced flow $\flowmap$ is volume-preserving.
\end{corollary}


\subsection{Structural properties of the \ac{FTRL} dynamics: the steep/non-steep dichotomy}

To proceed with our analysis, we will need to clarify the precise technical requirements for the dynamics' regularizer function $\hreg$.
These are as follows:
\begin{enumerate*}
[(\itshape i\hspace*{.5pt})]
\item
$\hreg\in C_{0}(\R_{+}^{\nPures}) \cap C_{2}(\R_{++}^{\nPures})$,
\ie $\hreg$ is continuous on $\R_{+}^{\nPures}$ and two times continuously differentiable on $\R_{++}^{\nPures}$;
\item
$\hreg$ is strongly convex on $\strats$;
and
\item
the inverse Hessian $\hmat(\strat) = \Hess(\hreg(\strat))^{-1}$ of $\hreg$ admits a Lipschitz extension to all of $\strats$ such that $\tanvec^{\top} \hmat(\strat) \tanvec > 0$ whenever $\supp(\tanvec) \supseteq \supp(\strat)$.
\end{enumerate*}
These conditions are purely technical in nature and 
they are satisfied by all the regularizers used in practice, \cf \citep{CT93,Kiw97b,ABB04,MS16,BecTeb03,SS11} and references therein.
As we discussed in the main body of our paper, there is an important distinction to be made depending on whether $\hreg$ is \emph{steep} or \emph{non-steep}.
Formally, we say that $\hreg$ is \emph{steep} if $\norm{\nabla\hreg(\strat)} \to \infty$ whenever $\strat\to\bd(\strats)$ and $\rank(\hmat(\point)) = \abs{\supp(\point)}$ for all $\strat$;
by contrast, if $\sup_{\strat\in\strats} \norm{\nabla\hreg(\strat)} < \infty$ and $\rank(\hmat(\point)) = \nPures$, we say that $\hreg$ is \emph{non-steep}.
The qualititative difference in behavior between these cases is illustrated in the figure below (which shows the very different behavior of the derivates of $\hreg$ near the boundary of the state space).


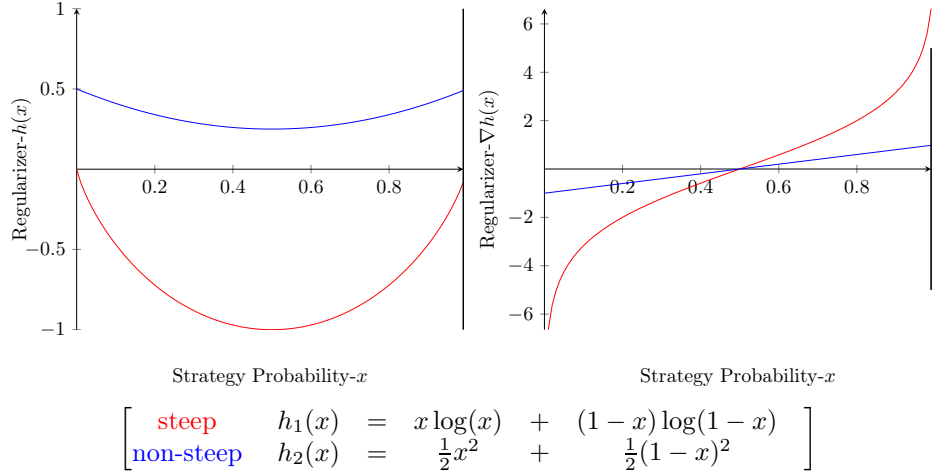
\begin{figure}[tbp]

\begin{tikzpicture}[scale=0.75]
\begin{axis}[
    axis lines = center,
    xlabel = {Strategy Probability-$x$},
    ylabel = {Regularizer-$h(x)$},
    x label style={at={(axis description cs:0.5,-0.1)},anchor=north},
    y label style={at={(axis description cs:-0.1,.5)},rotate=90,anchor=south}
]
\addplot [
    color=red, 
    domain=0:.99,
    samples=101
]{x*log2(x)+(1-x)*log2(1-x)};
\addplot [
    color=blue, 
    domain=0:.99, 
    samples=101
]{1/2*(x^2+(1-x)^2)};
\addplot [
thick,
color=black
] coordinates {(0.99, -1) (0.99, 1)};
\end{axis}
\end{tikzpicture}
\begin{tikzpicture}[scale=0.75]
\begin{axis}[
    axis lines = center,
    xlabel = {Strategy Probability-$x$},
    ylabel = {Regularizer-$\nabla h(x)$},
    x label style={at={(axis description cs:0.5,-0.1)},anchor=north},
    y label style={at={(axis description cs:-0.1,.5)},rotate=90,anchor=south}
]
\addplot [
    color=red, 
    domain=0:.99,
    samples=101
]{log2(x/(1-x))};
\addplot [
    color=blue, 
    domain=0:.99, 
    samples=101
]{2*x - 1};
\addplot [
thick,
color=black
] coordinates {(.99, -5) (.99, 5)};
\end{axis}
\end{tikzpicture}
\[
\begin{bmatrix}
\text{\color{red} steep }&h_1(x)&=&x\log(x)&+&(1-x)\log(1-x)&\\
\text{\color{blue} non-steep }&h_2(x)&=&\frac{1}{2}x^2&+&\frac{1}{2}(1-x)^2&
\end{bmatrix}
\]
\caption{Steep \vs non-steep regularizers (note in particular the singular behavior of the gradient at the boundary in the case of steep regularizers).}
\end{figure}



\begin{tcolorbox}[enhanced,width=5in, drop fuzzy shadow southwest, 
                    boxrule=0.4pt,sharp corners,colframe=yellow!80!black,colback=yellow!10]
The following lemma illustrates the relation between mixed strategies and score vectors and the mirror map \eqref{eq:choice} that defines the dynamics \eqref{eq:FTRL}.
We focus on the perspective of an arbitrary player, say $\play$, and for ease of notation we write $\mirror$, $\strat$ and $\dpoint$ instead of $\mirror_\play$, $\point_\play$ and $\dpoint_\play$ respectively.
The lemma begins to illustrate the gulf between the steep and non-steep cases.
\end{tcolorbox}

\begin{lemma}
\label{lem:KKT}
$\point=\mirror(\dpoint)$ if and only if there exist $\mu\in \R$ and $\nu_\pure\in\R_{+}$ such that, for all $\pure\in \pures$, we have:
\begin{enumerate*}
[\itshape a\upshape)]
\item
$\dpoint_\pure=\frac{\partial \hreg}{\partial \point_\pure}+\mu-\nu_\pure$;
and
\item
$\point_\pure \nu_\pure=0$
\end{enumerate*}
In particular, if $\hreg$ is steep, we have $\nu_\pure = 0$ for all $\pure\in\pures$.
\end{lemma}

\begin{proof}
Recall that 
\begin{align*}
\mirror(\dpoint)&= \argmax_{\point\in\points} \left\{ \braket{\dpoint}{\point} - \hreg(\point) \right\}\\
                &=\argmax \setdef*{\sum_{\pure\in\pures}\dpoint_\pure \point_\pure-\hreg(\point)}{\sum_{\pure\in\pures}\point_\pure=1 \mbox{ and } \forall \pure \in\pures:\point_\pure\geq 0 }
\end{align*}
The result follows by applying the \ac{KKT} conditions to this optimization problem and noting that, since the constraints are affine, the \ac{KKT} conditions are sufficient for optimality.
Our Langragian is   \[\mathcal{L}(\point,\mu,\nu)=( \sum_{\pure\in\pures}\dpoint_\pure \point_\pure - \hreg(\point) )-{\mu(\sum_{\pure\in\pures} \point_\pure-1)} + {\sum_{\pure\in\pures} \nu_\pure \point_\pure}\]
where the set of constraints (i) of the statement of the lemma  are the stationarity constraints, 
which in our case are $\nabla\mathcal{L}(\point,\mu,\nu)=0\Leftrightarrow
\nabla( \sum_{\pure\in\pures}\dpoint_\pure \point_\pure - \hreg(\point) )=\mu\nabla(\sum_{\pure\in\pures} \point_\pure-1) - \sum_{\pure\in\pures} \nu_\pure\nabla \point_\pure$
, while the set of constraints (ii) of the statement of the lemmas are the complementary slackness constraints.
Note that complementary slackness implies that whenever $\nu_\pure>0$ whenever $\pure\notin \supp(\point)$.
Finally, if $\hreg$ is steep, we have $\abs{\pd_{\pure}\hreg(\point)} \to \infty$ as $\point\to\bd(\points)$, which implies that the \ac{KKT} conditions admit a solution with $\nu_{\pure} = 0$.
\end{proof}

\begin{tcolorbox}[enhanced,width=5in, drop fuzzy shadow southwest, 
                    boxrule=0.4pt,sharp corners,colframe=yellow!80!black,colback=yellow!10]
The following lemma shows that if the support of $\point(t)$ does not change over a given interval of time, then the evolution of the players' mixed strategies under \eqref{eq:FTRL} follows a certain differential equation that can be calculated explicitly.
The lemma below also shows that the  trajectory of play coincides with the trajectory that would have resulted if the game were constrained to the strategies present in the support of $\orbit{\run}$.
Again, for ease of notation  we focus on  player $\play$ and omit $\play$ from all subscripts.
\end{tcolorbox}

\begin{proposition}
\label{prop:FTRL-strat}
Let $\orbit{\run} = \mirror(\dorbit{\run})$ be a mixed strategy orbit of \eqref{eq:FTRL}, and let $\runs$ be an interval over which $\supp(\orbit{\run})$ is constant.
Then, for all $\run\in\runs$, $\orbit{\run}$ satisfies the mixed strategy dynamics:
\begin{equation}
\label{eq:FTRL-strat}
\tag{\ref*{eq:FTRL}-s}
\dot\strat_{\pure}
	= \insum_{\purealt\in\supp(\strat)}
		\bracks*{\hmat_{\pure\purealt}(\strat) - \hmat_{\pure}(\strat) \hmat_{\purealt}(\strat)}\,
		\payv_{\pure}(\strat),
\end{equation}
where $\hmat_{\pure}(\strat) = \bracks*{\sum_{\purealt,\alt\purealt\in\supp(\strat)} \hmat_{\purealt\alt\purealt}(\strat)}^{-1/2} \sum_{\purealt\in\supp(\strat)} \hmat_{\pure\purealt}(\strat)$.
In particular, we have the following dichotomy:
\begin{enumerate}
\item
If $\hreg$ is steep, the dynamics \eqref{eq:FTRL-strat} are \emph{well-posed}, \ie they admit unique global solutions from any initial condition $\strat\in\strats$ \textpar{including the boundary}.
Moreover, the faces of $\strats$ are forward-invariant under \eqref{eq:FTRL-strat}:
the support of $\orbit{\run}$ remains constant for all $\run\geq0$.
\item
If $\hreg$ is non-steep, the dynamics \eqref{eq:FTRL-strat} are not well-posed:
solutions $\orbit{\run}$ to \eqref{eq:FTRL-strat} exist only up to a finite time, after which the support of $\orbit{\run}$ may change.
\end{enumerate}
\end{proposition}


\begin{proof}
For the first part of the lemma, we follow a line of reasoning due to \citep{MS16}.
Specifically, letting $g_{\pure}(\point) = \pd_{\pure}\hreg(\point)$, \cref{lem:KKT} yields
\begin{equation}
\dpoint_\pure(t)
	= g_\pure(t)+\mu(t), \,\,\forall t\in I,\forall\pure\in\pures^*
\end{equation}
Since $\dpoint_\pure$ and $g_\pure$ are both smooth, so is  $\mu(t)$.
Thus, differentiating with respect to $t$  we get
\begin{align*}
    \dot{\dpoint}_\pure(t)&=\sum_{\purealt\in\pures}\frac{\partial^2\hreg}{\partial\point_{\purealt}\partial\point_\pure} \dot{\point}_{\purealt}(t)+\dot{\mu}(t)\\
    &=\sum_{\purealt\in\pures^*}\frac{\partial^2\hreg}{\partial\point_{\purealt}\partial\point_\pure} \dot{\point}_{\purealt}(t)+\dot{\mu}(t)
\end{align*}
since for all $t\in I$ and $\purealt \in\pures\setminus\pures^*$, $\point_{\purealt}(t)=0$, and thus $\dot{\point}_\pure(t)=0$. Multiplying with the inverse of the Hessian, and omitting $t$ for brevity, we get
\begin{equation}\label{eq:dynamics_boundary_nonsteep}\dot{\point}_\pure=\sum_{\purealt\in\pures^*}\hmat_{\pure\purealt}\dot{\dpoint}_{\purealt}+
\sum_{\purealt\in\pures^*}\hmat_{\pure\purealt} \dot{\mu}
\end{equation}
By the definition of the dynamics, $\dot{\dpoint}_{\purealt}=\payv_{\purealt}$ and since the support remains constant $\sum_{\pure\in\pures^*}\dot{\point}_\pure=0$. Summing up \Cref{eq:dynamics_boundary_nonsteep} for $\pure\in\pures^*$ we get
\begin{equation}
\sum_{\pure,\purealt\in\pures^*}\hmat_{\pure\purealt}\payv_{\purealt}+G\dot{\mu}=0\Leftrightarrow \dot{\mu}=-\frac{\sum_{\purealt\in\pures^*}\hmat_{\purealt}\payv_{\purealt}}{G}
\end{equation}
where $G = \sum_{\purealt,\alt\purealt\in\supp(\strat)} \hmat_{\purealt\alt\purealt}(\strat)$.
Substituting the latter and $\dot{\dpoint}_{\purealt}=\payv_{\purealt}$ to \Cref{eq:dynamics_boundary_nonsteep} we get the desired result.

For the second part of the lemma, the well-posedness of \eqref{eq:FTRL-strat} follows from the fact that $\hmat(\point)$ admits a Lipschitz continuous extension to all of $\point$;
moreover, by the rank assumption, the field $\hmat_{\pure}(\point)$ is also Lipschitz continuous (since the denominator does not vanish; recall that $\im\hmat(\point) = \R^{\supp(\point)}$).
Finally, forward invariance follows by noting that, for every initial condition $\strat\in\strats$, the quantity $\sum_{\pure\in\supp(\point)}\orbit[\pure]{\run}$ is a constant of motion (identically equal to $1$), and that $\dot\point_{\pure} = 0$ whenever $\point_{\pure} = 0$.
\end{proof}

\begin{tcolorbox}[enhanced,width=5in, drop fuzzy shadow southwest, 
                    boxrule=0.4pt,sharp corners,colframe=yellow!80!black,colback=yellow!10]
A key take-away from the above result is that, in the steep case, the dynamics \eqref{eq:FTRL-strat} can be seen as a patchwork of dynamical systems, each evolving on a specific face of $\strat$, and each a continuous extension of the other at the points where they come into contact.
Neither of the above is true for the non-steep case.
\end{tcolorbox}

\subsection{Volume preservation in $\dpoints$ and $\intstrats$}

\hspace{3cm }

\begin{tcolorbox}[enhanced,width=5in, drop fuzzy shadow southwest, 
                    boxrule=0.4pt,sharp corners,colframe=yellow!80!black,colback=yellow!10]

After restating it, we proceed to show \cref{prop:volume} by simply  applying Liouville's formula (\cref{thm:Liouville}).
\end{tcolorbox}

\volumeY*

\begin{proof}

We have to show that the dynamics of \eqref{eq:FTRL} (i.e., $\dotdorbit{\run}
	= \payv(\mirror(\dorbit{\run}))$)  are incompressible.
For any player $\play$ and any $\pure\in\pures_\play$ we have
\begin{equation}
\label{eq:incompressible}
\frac{\pd \payv_{\play\pure}}{\pd \dpoint_{\play\pure}}
	= \sum_{\purealt\in\pures_{\play}}
		\frac{\pd  \payv_{\play\pure_{\play}}}{\pd \strat_{\play\purealt}}
		\frac{\pd \strat_{\play\purealt}}{\pd \dpoint_{\play\pure}}
	= 0,
\end{equation}
because $\payv_{\play}$ does not depend on $\strat_{\play}$.
We thus obtain $\operatorname{div}_{\dpoint} \payv(\dpoint) = 0$, i.e., the dynamics \eqref{eq:FTRL} are incompressible.
The result then follows from Liouville's formula
\end{proof}

\begin{tcolorbox}[enhanced,width=5in, drop fuzzy shadow southwest, 
                    boxrule=0.4pt,sharp corners,colframe=yellow!80!black,colback=yellow!10]

In the following lemma, we show that the flow defined by \ac{FTRL} in the interior $\intstrats$ of $\points$  is incompressible under a suitably defined measure. For that, using Liouville's formula, we first show that in the so-called $\zpoint$-space, i.e., a ``slice'' of the payoff space, the respective flow is incompressible. Using that, we can easily get a diffeomorphism from  $\intstrats$ to the $\zpoint$-space, we define the volume of a set in $\intstrats$  to be the volume of the corresponding set in the $\zpoint$-space, and thus incompresssibility in the interior  comes for free.
\end{tcolorbox}


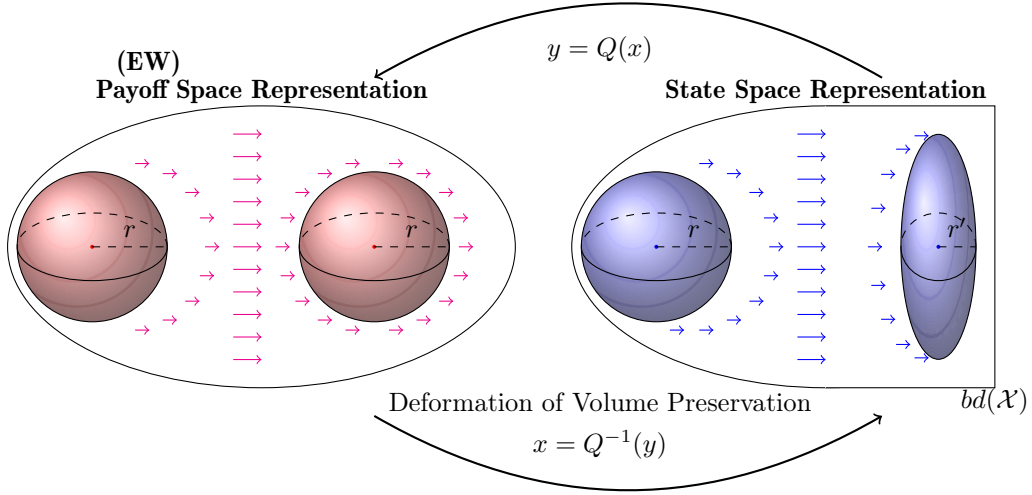
\begin{figure}[htbp]
\begin{tikzpicture}[scale=0.75]
\def\radius{1.33}
\def\centerAx{-5}
\def\centerAy{0}
\def\centerA{(\centerAx,\centerAy)}
 \fill (-4,3.25) node {\textbf{(EW)}};
  \fill (4,-2.75) node {Deformation of Volume Preservation};
 \draw (-2,2.75) node {\textbf{Payoff Space Representation}};
 \draw (-2,0) ellipse (4.5cm and 2.5cm);
  \shade[ball color = red!40, opacity = 0.4] \centerA circle ({\radius});
  \draw \centerA circle ({\radius});
  \draw ($ (-{\radius},0) + (\centerAx,\centerAy) $) arc (180:360:{\radius} and 0.6);
  \draw[dashed] ($ ({\radius},0) + (\centerAx,\centerAy) $) arc (0:180:{\radius} and 0.6);
  \fill[fill=red] \centerA circle (1pt);
  \draw[dashed] \centerA -- node[above]{$r$}  ($ ({\radius},0) + (\centerAx,\centerAy) $);
\foreach \t in {-80,-60,...,80}
{
		\draw[magenta,->] ({1.5*cos(\t)-4.5},{1.5*sin(\t)})--({1.5*cos(\t)-4.25},{1.5*sin(\t)});
		\draw[magenta,->] ({1.5*cos(\t)},{1.5*sin(\t)})--({1.5*cos(\t)+.25},{1.5*sin(\t)});
		\draw[magenta,->] (-{1.5*cos(\t)-.25},{1.5*sin(\t)})--(-{1.5*cos(\t)},{1.5*sin(\t)});
}

\foreach \t in {-10,-8,...,10}
{
		\draw[magenta,->] (0-2.5,{\t/5})--(0.5-2.5,{\t/5});
}
\def\centerBx{0}
\def\centerBy{0}
\def\centerB{(\centerBx,\centerBy)}
   \shade[ball color = red!40, opacity = 0.4] \centerB circle ({\radius});
  \draw \centerB circle ({\radius});
  \draw ($ (-{\radius},0) + (\centerBx,\centerBy) $) arc (180:360:{\radius} and 0.6);
  \draw[dashed] ($ ({\radius},0) + (\centerBx,\centerBy) $) arc (0:180:{\radius} and 0.6);
  \fill[fill=red] \centerB circle (1pt);
  \draw[dashed] \centerB -- node[above]{$r$}  ($ ({\radius},0) + (\centerBx,\centerBy) $);

\def\trans{10}
\def\centerCx{0+\trans/2}
\def\centerCy{0}
\def\centerC{(\centerCx,\centerCy)}
 \draw (-2+\trans,2.75) node {\textbf{State Space Representation}};
  \draw (-2+\trans,0) ellipse (4.5cm and 2.5cm);
  \filldraw[white] (\trans-2,2.5) -- (\trans*2-2,2.5) -- (\trans*2-2,-2.5) -- (\trans-2,-2.5) -- (\trans-2,2.5) ;
  \draw[black] (\trans-2,2.5) -- (\trans*1.3-2,2.5) -- (\trans*1.3-2,-2.5) -- (\trans-2,-2.5) ;
  
  \shade[ball color = blue!40, opacity = 0.4] \centerC circle ({\radius});
  \draw \centerC circle ({\radius});
  \draw ($ (-{\radius},0) + (\centerCx,\centerCy) $) arc (180:360:{\radius} and 0.6);
  \draw[dashed] ($ ({\radius},0) + (\centerCx,\centerCy) $) arc (0:180:{\radius} and 0.6);
  \fill[fill=blue] \centerC circle (1pt);
  \draw[dashed] \centerC -- node[above]{$r$}  ($ ({\radius},0) + (\centerCx,\centerCy) $);
\foreach \t in {-100,-80,...,80}
{
		\draw[blue,->] ({1.5*cos(\t)-4.5+\trans},{1.5*sin(\t)})--({1.5*cos(\t)-4.25+\trans},{1.5*sin(\t)});
}
\foreach \t in {-80,-60,...,80}
{
	\draw[blue,->] (-{cos(\t)-.25+\trans},{2*sin(\t)})--(-{cos(\t)+\trans},{2*sin(\t)});
}
\foreach \t in {-10,-8,...,10}
{
		\draw[blue,->] (0-2.5+\trans,{\t/5})--(0.5-2.5+\trans,{\t/5});
}

\def\centerDx{5+\trans/2}
\def\centerDy{0}
\def\centerD{(\centerDx,\centerDy)}
  \shade[ball color = blue!40, opacity = 0.4] \centerD ellipse ({0.5*\radius} and {1.5*\radius});
  \draw  \centerD ellipse ({0.5*\radius} and {1.5*\radius});
  \draw[dashed] ($({\radius/2},0) +  (\centerDx,\centerDy) $) arc (0:180:{0.5*\radius} and 0.6);
  \draw ($ (-{\radius/2},0) + (\centerDx,\centerDy) $) arc (180:360:{0.5*\radius} and 0.6);
  \fill[fill=blue] \centerD circle (1pt);
  \draw[dashed] \centerD -- node[above]{$r'$}  ($ ({\radius/2},0) + (\centerDx,\centerDy) $);

\draw [<-,thick]  (0,3) to [out=30,in=150] (-1+\trans,3) ;
\draw [->,thick]  (0,-3) to [out=-30,in=-150] (-1+\trans,-3);
\fill (4,3.5) node {$y=Q(x)$};
\fill (4,-3.5) node {$x=Q^{-1}(y)$};
\fill (\trans*1.3-2,-2.75) node {$bd(\mathcal{X})$};

\end{tikzpicture}
\caption{From payoffs to strategies and back, the steep case:
deformation of neighborhoods under the logit map $\mirror = \logit$.}
\end{figure}


\begin{lemma}\label{lem:X_volume_preserving}
There exists a measure $\mu_\point$ for which the flow in the interior of $\points$ is incompressible, i.e., for any subset $U\subset\intstrats$ of initial conditions, and any $t_0\geq 0$ so that for any $0\leq t\leq t_0$: $\flowmap(U,t)\subset\intstrats$, it is  $\mu_\point(U)=\mu_\point(\flowmap(U,t))$. 
\end{lemma}

\begin{proof}
First we go on to define the $\zpoint$-space. The intuition for defining and using the $\zpoint$-space  can be based on \cref{lem:KKT} which implies that 
for any $\point\in\intstrats$,  any two corresponding points $\dpoint,\dpoint'$   in the $\dpoint$-space differ by a constant, since for all $\play$ and $\pure_j\in\pures_\play$, $\dpoint_{\play\pure_j}=\frac{\partial \hreg}{\partial \point_{\play\pure_j}}+\mu_\play$ and $\dpoint'_{\play\pure_j}=\frac{\partial \hreg}{\partial \point_{\play\pure_j}}+\mu_\play'$ for some $\mu_\play$ and $\mu_\play'$ (recall $\point\in\intstrats$ implies $\nu_{\play\pure_j}=0$). Thus, all $\dpoint$'s that correspond to an $\point\in\intstrats$ form an equivalent class.  For each class, we pick as representative the $\dpoint$ in the class that has $0$ in some specific  coordinate $\purebench_\play$, for every player $\play$. The set of representatives form the $\zpoint$-space and there is a a one to one correspondence of points of $\intstrats$ to points in the $\zpoint$-space which moreover can be used to define an incompressible flow in  $\intstrats$.


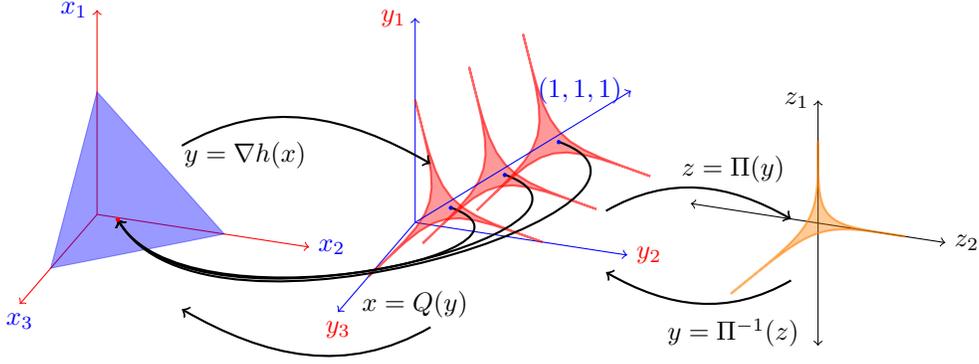
\begin{figure}[htbp]
\tdplotsetmaincoords{65}{110}
\begin{tikzpicture}[tdplot_main_coords,scale=0.6]
  \def\laxis{5}
  \def\ltriangle{3}
  \begin{scope}[->,red]
    \draw (0,0,0) -- (\laxis,0,0) node [below] {\textcolor{blue}{$x_3$}};
    \draw (0,0,0) -- (0,\laxis,0) node [right] {\textcolor{blue}{$x_2$}};
    \draw (0,0,0) -- (0,0,\laxis) node [left] {\textcolor{blue}{$x_1$}};
  \end{scope}
  \filldraw [opacity=.4,blue] (\ltriangle,0,0) -- (0,\ltriangle,0) --
  (0,0,\ltriangle) -- cycle;
  
  \def\transferx{0}
  \def\transfery{7.5}
  \def\transferz{1}

  \def\laxis{5}
  \def\ltriangle{1.5}
  \begin{scope}[->,blue]
    \draw (0+\transferx,0+\transfery,0+\transferz) -- (\laxis+\transferx,0+\transfery,0+\transferz) node [below]{\textcolor{red}{$y_3$}};
    \draw (0+\transferx,0+\transfery,0+\transferz) -- (0+\transferx,\laxis+\transfery,0+\transferz) node [right] {\textcolor{red}{$y_2$}};
    \draw (0+\transferx,0+\transfery,0+\transferz) -- (0+\transferx,0+\transfery,\laxis+\transferz) node [left] {\textcolor{red}{$y_1$}};
  \end{scope}

\def\ltriangle{3}
\foreach \k in {1,3,5}
{

  \def\transferx{0+\k}
  \def\transfery{7.5+\k}
  \def\transferz{1+\k}
\filldraw [thick,draw opacity=.6,fill opacity=.4,red] 
    (  {\ltriangle*(0^3-(0^3+(1-0)^3)/3 ) +\transferx}, 
            {\ltriangle*((1-0)^3 -(0^3+(1-0)^3)/3) +\transfery} , 
            {\ltriangle*(0-(0^3+(1-0)^3)/3) +\transferz} )  
    \foreach \t in {0,0.1,...,1.1}  
    { -- (  {\ltriangle*(\t^3-(\t^3+(1-\t)^3)/3 ) +\transferx}, 
            {\ltriangle*((1-\t)^3 -(\t^3+(1-\t)^3)/3) +\transfery} , 
            {\ltriangle*(0-(\t^3+(1-\t)^3)/3) +\transferz} ) 
    }
    \foreach \t in {0,0.1,...,1.1}  
    { -- (  {\ltriangle*((1-\t)^3-(\t^3+(1-\t)^3)/3 ) +\transferx}, 
            {\ltriangle*(0 -(\t^3+(1-\t)^3)/3) +\transfery} , 
            {\ltriangle*(\t^3-(\t^3+(1-\t)^3)/3) +\transferz} ) 
    }
    \foreach \t in {0,0.1,...,1.1}  
    { -- (  {\ltriangle*(0-(\t^3+(1-\t)^3)/3 ) +\transferx}, 
            {\ltriangle*(\t^3 -(\t^3+(1-\t)^3)/3) +\transfery} , 
            {\ltriangle*((1-\t)^3-(\t^3+(1-\t)^3)/3) +\transferz} ) 
    };

\draw [->,thick] (   {(0.3)^3+\transferx}  ,  {(0.6)^3+\transfery} ,{( 0.2 )^3 +\transferz })to [out=-20,in=-60] (0.3,0.6,0.1) ;
  \filldraw[blue]  (   {(0.3)^3+\transferx}               ,      {(0.6)^3+\transfery}             ,{( 0.1 )^3 +\transferz }) circle (1pt);
  \filldraw[red] (0.3,0.6,0.1) circle (1pt);
    
}
\draw[->,blue] (0+\transferx,0+\transfery,0+\transferz) -- (8+\transferx,8+\transfery,8+\transferz) node [left] {\textcolor{blue}{$(1,1,1)$}};

\draw [->,thick]  (0,2,2) to [out=30,in=150] (-1,+\transfery,2) ;
\draw [<-,thick]  (0,2,-2) to [out=-30,in=-150] (-1,\transfery,-2);
\fill (0,3.5,2) node {$y=\nabla h(x)$};
\fill (0,\transfery,-1) node {$x=Q(y)$};

  \def\transferx{0}
  \def\transfery{17}
  \def\transferz{2.5}
   \def\laxis{3}
  \def\ltriangle{2}
  \begin{scope}[->,black]
    \draw (0+\transferx,0+\transfery,0+\transferz) -- (0+\transferx,\laxis+\transfery,0+\transferz) node [right] {\textcolor{black}{$z_2$}};
    \draw (0+\transferx,0+\transfery,0+\transferz) -- (0+\transferx,0+\transfery,\laxis+\transferz) node [left] {\textcolor{black}{$z_1$}};
   \draw (0+\transferx,0+\transfery,0+\transferz) -- (0+\transferx,-\laxis+\transfery,0+\transferz) node [right] {\textcolor{black}{}};
    \draw (0+\transferx,0+\transfery,0+\transferz) -- (0+\transferx,0+\transfery,-\laxis+\transferz) node [left] {\textcolor{black}{}};
  \end{scope}

\draw [->,thick]  (0,12,2) to [out=30,in=150] (-1,\transfery-1,2) ;
\draw [<-,thick]  (0,12,0.5) to [out=-30,in=-150] (-1,\transfery-1,0.5);
\fill (0,15,3.5) node {$z=\Pi(y)$};
\fill (0,15,-0.5) node {$y=\Pi^{-1}(z)$};
  
\filldraw [thick,draw opacity=.6, fill opacity=.4, orange]  
(\transferx, { \ltriangle*(0^3-(1-0-(1-0))^3) +\transfery} , { \ltriangle*((1-0)^3-(1-0-(1-0))^3) +  +\transferz} ) 
\foreach \t in {0,0.01,...,1.01}
{ 
    --(\transferx, { \ltriangle*(\t^3-(1-\t-(1-\t))^3) +\transfery} , { \ltriangle*((1-\t)^3-(1-\t-(1-\t))^3) +  +\transferz} ) 
}
\foreach \t in {0,0.01,...,1.01}
{ 
    --(\transferx, { \ltriangle*((1-\t)^3-(1-(1-\t))^3) +\transfery} , { \ltriangle*(-(1-(1-\t))^3) +  +\transferz} ) 
}
\foreach \t in {0,0.01,...,1.0}
{ 
    --(\transferx, { \ltriangle*(-(1-\t)^3) +\transfery} , { \ltriangle*(\t^3-(1-\t)^3) +  +\transferz} ) 
}
;
\end{tikzpicture}
\caption{The three different representation spaces of \eqref{eq:FTRL}:
primal (strategies, left);
dual (payoffs/scores, center);
dual quotient (payoff/score differences, right).}
\end{figure}


So, for a \emph{benchmark} strategy $\purebench_\play\in\pures_{\play}$ for every player $\play\in\players$ and for all $\pure \in \pures_{\play} \exclude{\purebench_{\play}}\equiv\puresbench_\play$ consider the corresponding score differences
\begin{equation}
\label{eq:score-z}
\zpoint_{\play\pure}
	= \dpoint_{\play\pure} - \dpoint_{\play\purebench_\play}.%
\end{equation}
Obviously, $\zpoint_{\play\pure}
	= \dpoint_{\play\pure_\play} - \dpoint_{\play\purebench_\play}$ is identically zero so we can ignore it in the above definition.
In so doing, we obtain a linear map $\Pi_{\play}\from\R^{\pures_{\play}} \to \R^{\puresbench_{\play}}$ sending $\dpoint_{\play} \mapsto \zpoint_{\play}$;
aggregating over all players, we also write $\Pi$ for the product map $\Pi = (\Pi_{1},\dotsc,\Pi_{\nPlayers})$ sending $\dpoint\mapsto \zpoint$.
For posterity, note that this map is surjective but \emph{not} injective,%
\footnote{Specifically, $\Pi_{\play}(\dpoint_{\play}) = \Pi_{\play}(\dpoint_{\play}')$ if and only if $\dpoint_{\play\pure_{\play}}' = \dpoint_{\play\pure_{\play}} + c$ for some $c\in\R$ and all $\pure_{\play}\in\pures_{\play}$.}
so it does not allow us to recover the score vector $\dpoint$ from the score difference vector $\zpoint$.

Now, under \ac{FTRL}, the score differences \eqref{eq:score-z} evolve as
\begin{equation}
\label{eq:dyn-z}
\dot \zpoint_{\play\pure}
	= \payv_{\play\pure}(\strat(t)) - \payv_{\play\purebench_{\play}}(\strat(t)).
\end{equation}
Our first step below is to show that \eqref{eq:dyn-z} 
constitutes a well-defined dynamical system on $\zpoint$ as long as the correpsonding $\point$'s remain in $\intstrats$.

To do so, consider the reduced mirror map $\hat\mirror_{\play}\from\R^{\puresbench_{\play}} \to \strats_{\play}$ defined as
\begin{equation}
\label{eq:choice-z}
\hat\mirror_{\play}(\zpoint_{\play})
	= \mirror_{\play}(\dpoint_{\play})
\end{equation}
for some $\dpoint_{\play}\in\R^{\pures_{\play}}$ such that $\Pi_{\play}(\dpoint_{\play}) = \zpoint_{\play}$.
That such a $\dpoint_{\play}$ exists is a consequence of $\Pi_{\play}$ being surjective;
furthemore, that $\hat\mirror_{\play}(\zpoint_{\play})$ is well-defined is a consequence of the fact that $\mirror_{\play}$ is invariant on the fibers of $\Pi_{\play}$.
Indeed, by construction, and as long as the corresponding $\point$'s remain in $\intstrats$ we have $\Pi_{\play}(\dpoint_{\play}) = \Pi_{\play}(\dpoint_{\play}')$ if and only if $\dpoint_{\play\pure}' = \dpoint_{\play\pure} + c$ for some $c\in\R$ and all $\pure\in\pures_{\play}$.
Hence, by the definition of $\mirror_{\play}$, we get
\begin{flalign}
\mirror_{\play}(\dpoint_{\play}')
	&\txs
	= \argmax\limits_{\strat_{\play}\in\strats_{\play}} \braces*{\braket{\dpoint_{\play}}{\strat_{\play}} + c \sum_{\pure\in\pures_{\play}} \strat_{\play\pure} - \hreg_{\play}(\strat_{\play})}
	\notag\\
	&= \argmax_{\strat_{\play}\in\strats_{\play}} \braces{\braket{\dpoint_{\play}}{\strat_{\play}} - \hreg_{\play}(\strat_{\play})}
	= \mirror_{\play}(\dpoint_{\play}),
\end{flalign}
where we used the fact that $\sum_{\pure\in\pures_{\play}} \strat_{\play\pure} = 1$.
The above shows that $\mirror_{\play}(\dpoint_{\play}') = \mirror_{\play}(\dpoint_{\play})$ if and only if $\Pi_{\play}(\dpoint_{\play}) = \Pi_{\play}(\dpoint_{\play}')$, so $\hat\mirror_{\play}$ is well-defined.
Letting $\hat\mirror \equiv (\hat\mirror_{1},\dotsc,\hat\mirror_{\nPlayers})$ denote the aggregation of the players' individual mirror maps $\hat\mirror_{\play}$, it follows immediately that $\mirror(y) = \hat\mirror(\Pi(y)) = \hat\mirror(z)$ by construction.

Hence, the dynamics \eqref{eq:dyn-z} may be written as
\begin{equation}
\label{eq:FRL-z}
\dot \zpoint
	= V(\zpoint),
\end{equation}
where
\begin{equation}
\nu_{\play\pure}(\zpoint)
	= \payv_{\play\pure}(\hat\mirror_{\play}(\zpoint)) - \payv_{\play\purebench_{\play}}(\hat\mirror_{\play}(\zpoint)).
\end{equation}
These dynamics obviously constitute an autonomous system. 

Next, we show  incompressibiity of the $\zpoint$-space. 
Indeed, for all $\pure\in\pures$ we have
\begin{equation}
\label{eq:incompressible}
\frac{\pd \nu_{\play\pure}}{\pd \zpoint_{\play\pure}}
	= \sum_{\purealt\in\pures_{\play}}
		\frac{\pd  \nu_{\play\pure_{\play}}}{\pd \strat_{\play\purealt}}
		\frac{\pd \strat_{\play\purealt}}{\pd \zpoint_{\play\pure}}
	= 0,
\end{equation}
because $\payv_{\play}$ does not depend on $\strat_{\play}$.
We thus obtain $\nabla_{\zpoint} \cdot V(\zpoint) = 0$, i.e., the dynamics \eqref{eq:FRL-z} are incompressible.

For the last step, for a set $A\subset\intstrats$ define $\mu_\point(A):=\mu_\zpoint(\mirror^{-1}(A))$,  where   $\mu_\zpoint$ is the Lebesgue measure in the $\zpoint$-space. Then for any $U\subset\intstrats$, as long as $\flowmap(U,t)$ remains in $\intstrats$, 
it is 
$$\mu_\point(U)=\mu_\zpoint(\mirror^{-1}(U))=\mu_\zpoint(\mirror^{-1}(\flowmap(U,t)))=\mu_\point(\flowmap(U,t))$$ 
as needed.
\end{proof}


\section{Proof of \cref{thm:unstable-int}}
\label{app:interior}


\begin{tcolorbox}[enhanced,width=5in, drop fuzzy shadow southwest, 
                    boxrule=0.4pt,sharp corners,colframe=yellow!80!black,colback=yellow!10]

Below we  show that there are no asymptotically stable sets (or points) in $\intstrats$. Indeed, if this were the case, there would be a full-measure set of initial conditions outside the asymptotically stable set $A^*$ that converges to $A^*$, while at the same time its trajectories are bounded (by stability).
This contradicts Poincaré's recurrence theorem, because the flow in $\intstrats$ is volume-preserving by \cref{lem:X_volume_preserving}.
\end{tcolorbox}

\begin{theorem}
\label{thm:No_Asset_Interior}
Let $A^*$ be a closed set of $\intstrats$. Then $A^*$ is not asymptotically stable under \ac{FTRL}.
\end{theorem}

\begin{proof}
To reach a contradiction, let $A^*$  be an asymptotically stable set, \ie attracting and Lyapunov stable, belonging in $\intstrats$. 
Since $A^*$ is attracting, there exists a neighborhood $U$ of $A^*$ all points of which converge to $A^*$.
Without loss of generality, since $A^*$ is closed, we may assume that $U$ lies in  $\intstrats$, and its closure is disjoint from the boundary of $\strats$.

Now, since $A^*$ is Lyapunov stable, there exists some neighborhood $U_0$ of $A^*$ so that whenever $\point(0)\in U_0$, $\point(t)\in U$.  Pick some $\point_0\in U_0\setminus A^*$.
Since $A^*$ is closed and $U_0$ is open, there is a small enough neighborhood $E$ of $\point_0$ so that  all  points of $E$ 
lie inside $U_0\setminus A^*$. 
By Lyapunov stability, 
for all $t\geq 0$,  $\flowmap(E,t)\subseteq U$. But then  the set $E_\infty = \cup_{t\geq 0} \flowmap( E,t)$ is bounded, having positive measure that does not change over time 
(\Cref{lem:X_volume_preserving}). 
Therefore it is a Poincaré recurrent set. But this means that all but a measure zero set of initializations in 
$E$ lead to recurrent trajectories that return infinitely often to 
$E$. Picking $E$ 
to be bounded away from $A^*$ (which is a closed set) we conclude that there are points in $E$ (and thus $U$) that do not converge to  $A^*$, a contradiction.
\end{proof}


\begin{figure}[tbp]
\begin{tikzpicture}[scale=0.7]
    \coordinate (zstar) at (-2, 0);
    \coordinate (zstar2) at (-1.2, -0.5);
    \coordinate (z0) at (2, 0);
    \draw [red,thick,domain=10:149, ->] plot ({-1.5 + 3.5 * cos(\x)}, {0.5+3 * sin(\x)});
    \draw [red,thick,domain=200:320, ->] plot ({-1.2 + 3.5 * cos(\x)}, {1+3.5 * sin(\x)});
    \draw (0,0) ellipse (7 and 4) ;
    \draw (0,4) node[above] {$U$};
    \draw[dashdotted] (-1,-0.2) ellipse (5 and 3) ;
    \draw (-1,2.8) node[above] {$U_0$};
    \draw[color = white , name path=Dz] (zstar) circle (2.4);
    \draw (zstar) circle (1);
    \draw (zstar) node[above] {$A^*$};
    \path [draw, ->, bend left, name path=A--B] (z0) to (zstar2);
    \path [name intersections={of=A--B and Dz, by=E}];
    \filldraw (z0) circle (1pt);
    \draw (2,0.1) node[above] {$x_0$};
    \filldraw[blue,  opacity=0.2] (z0) circle (1.2);
    \draw (2,-1.9) node[above] {$E$};
    \path let \p1 = (E)  in coordinate (Eprime) at (-140 -\x1,-\y1);
    \filldraw[blue,  opacity=0.2] (Eprime) circle (1.2);
    \draw (Eprime) node[above] {$\Phi (E, t)$};
    \node[draw] at (-0.9, 1.9) {Poincar\'e Reccurence};
\end{tikzpicture}
\caption{The various sets in the proof of \cref{thm:No_Asset_Interior}.}
\end{figure}
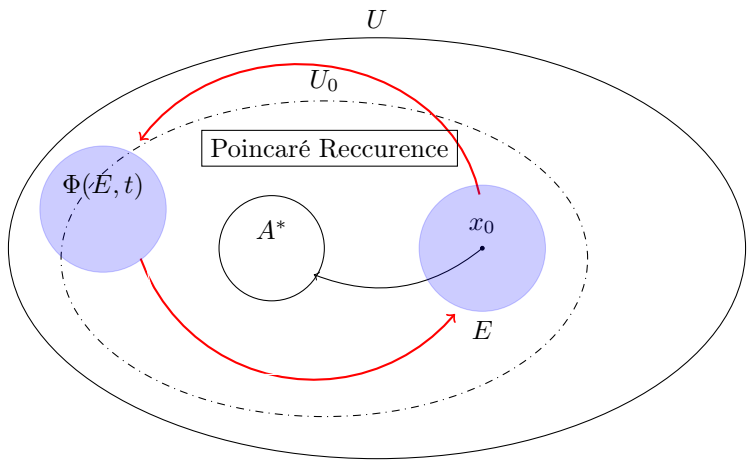


Now, given that any singleton set $\{\strat\}$, $\strat\in\strats$, is closed, the above yields:

\begin{theorem}
\label{thm:No_Aspoint_Interior}
There are no asymptotically stable points in $\intstrats$ 
\end{theorem}

\Cref{thm:unstable-int} (restated below) then follows as a corollary.

\stableInt*

\section{Proof of \cref{thm:unstable-mixed}: the non-steep case}
\label{app:nonsteep}

Our goal in this appendix is to provide the proof of \cref{thm:unstable-mixed}, which we restate below for convenience:

\stableMixed*

Because of the fundamental dichotomy between steep and non-steep \ac{FTRL} dynamics, we will break the proof in two cases, treating here the non-steep regime;
the steep case will be proved in \cref{app:steep} as a consequence of a more general result.
The fundamental distinction between the two cases is that, in the non-steep regime, the mixed-strategy dynamics of \eqref{eq:FTRL} could change support infinitely many times, which means that the type of volume-preservation arguments employed in the previous section cannot work (because the corresponding preimages in the $z$-space could have infinite volume; see below for a graphical illustration).
However, as we show below, this ``change of support'' is a blessing in disguise:
if $\eq$ is asymptotically stable, nearby trajectories will end up employing only those strategies present in $\eq$ in finite time.


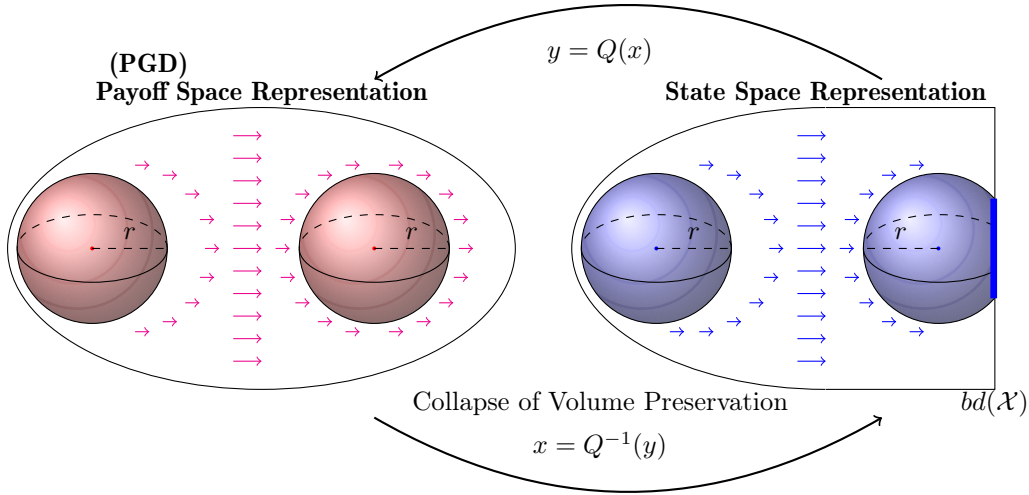
\begin{figure}[htbp]
\begin{tikzpicture}[scale=0.75]
\def\radius{1.33}
\def\centerAx{-5}
\def\centerAy{0}
\def\centerA{(\centerAx,\centerAy)}
 \fill (-4,3.25) node {\textbf{(PGD)}};
 \fill (4,-2.75) node {Collapse of Volume Preservation};
 \draw (-2,2.75) node {\textbf{Payoff Space Representation}};
 \draw (-2,0) ellipse (4.5cm and 2.5cm);
  \shade[ball color = red!40, opacity = 0.4] \centerA circle ({\radius});
  \draw \centerA circle ({\radius});
  \draw ($ (-{\radius},0) + (\centerAx,\centerAy) $) arc (180:360:{\radius} and 0.6);
  \draw[dashed] ($ ({\radius},0) + (\centerAx,\centerAy) $) arc (0:180:{\radius} and 0.6);
  \fill[fill=red] \centerA circle (1pt);
  \draw[dashed] \centerA -- node[above]{$r$}  ($ ({\radius},0) + (\centerAx,\centerAy) $);
\foreach \t in {-80,-60,...,80}
{
		\draw[magenta,->] ({1.5*cos(\t)-4.5},{1.5*sin(\t)})--({1.5*cos(\t)-4.25},{1.5*sin(\t)});
		\draw[magenta,->] ({1.5*cos(\t)},{1.5*sin(\t)})--({1.5*cos(\t)+.25},{1.5*sin(\t)});
		\draw[magenta,->] (-{1.5*cos(\t)-.25},{1.5*sin(\t)})--(-{1.5*cos(\t)},{1.5*sin(\t)});
}

\foreach \t in {-10,-8,...,10}
{
		\draw[magenta,->] (0-2.5,{\t/5})--(0.5-2.5,{\t/5});
}
\def\centerBx{0}
\def\centerBy{0}
\def\centerB{(\centerBx,\centerBy)}
   \shade[ball color = red!40, opacity = 0.4] \centerB circle ({\radius});
  \draw \centerB circle ({\radius});
  \draw ($ (-{\radius},0) + (\centerBx,\centerBy) $) arc (180:360:{\radius} and 0.6);
  \draw[dashed] ($ ({\radius},0) + (\centerBx,\centerBy) $) arc (0:180:{\radius} and 0.6);
  \fill[fill=red] \centerB circle (1pt);
  \draw[dashed] \centerB -- node[above]{$r$}  ($ ({\radius},0) + (\centerBx,\centerBy) $);

\def\trans{10}
\def\centerCx{0+\trans/2}
\def\centerCy{0}
\def\centerC{(\centerCx,\centerCy)}
 \draw (-2+\trans,2.75) node {\textbf{State Space Representation}};
  \draw (-2+\trans,0) ellipse (4.5cm and 2.5cm);
  \filldraw[white] (\trans-2,2.5) -- (\trans*2-2,2.5) -- (\trans*2-2,-2.5) -- (\trans-2,-2.5) -- (\trans-2,2.5) ;
  \draw[black] (\trans-2,2.5) -- (\trans*1.3-2,2.5) -- (\trans*1.3-2,-2.5) -- (\trans-2,-2.5) ;
  
  \shade[ball color = blue!40, opacity = 0.4] \centerC circle ({\radius});
  \draw \centerC circle ({\radius});
  \draw ($ (-{\radius},0) + (\centerCx,\centerCy) $) arc (180:360:{\radius} and 0.6);
  \draw[dashed] ($ ({\radius},0) + (\centerCx,\centerCy) $) arc (0:180:{\radius} and 0.6);
  \fill[fill=blue] \centerC circle (1pt);
  \draw[dashed] \centerC -- node[above]{$r$}  ($ ({\radius},0) + (\centerCx,\centerCy) $);
\foreach \t in {-100,-80,...,80}
{
		\draw[blue,->] ({1.5*cos(\t)-4.5+\trans},{1.5*sin(\t)})--({1.5*cos(\t)-4.25+\trans},{1.5*sin(\t)});
}
\foreach \t in {-80,-60,...,80}
{
	\draw[blue,->] (-{1.5*cos(\t)-.25+\trans},{1.5*sin(\t)})--(-{1.5*cos(\t)+\trans},{1.5*sin(\t)});
}
\foreach \t in {-10,-8,...,10}
{
		\draw[blue,->] (0-2.5+\trans,{\t/5})--(0.5-2.5+\trans,{\t/5});
}

\def\centerDx{5+\trans/2}
\def\centerDy{0}
\def\centerD{(\centerDx,\centerDy)}
  \shade[ball color = blue!40, opacity = 0.4] \centerD circle ({\radius});
  \draw \centerD circle ({\radius});
  \draw ($ (-{\radius},0) + (\centerDx,\centerDy) $) arc (180:360:{\radius} and 0.6);
  \draw[dashed] ($ ({\radius},0) + (\centerDx,\centerDy) $) arc (0:180:{\radius} and 0.6);
  \fill[fill=blue] \centerD circle (1pt);
  \draw[dashed] \centerD -- node[above]{$r$}  ($ ({-\radius},0) + (\centerDx,\centerDy) $);
  \filldraw[white]
 (\centerDx-0.3+\radius,1) -- (\centerDx+0.1 +\radius,1) -- (\centerDx+0.1 +\radius,-1) -- (\centerDx-0.3+\radius,-1) -- (\centerDx-0.3+\radius,1) ;
   \filldraw[blue]
 (\centerDx-0.4+\radius,0.875) -- (\centerDx-0.3 +\radius,0.875) -- (\centerDx-0.3+\radius,-0.875) -- (\centerDx-0.4+\radius,-0.875) -- (\centerDx-0.4+\radius,0.875) ;
 

\draw [<-,thick]  (0,3) to [out=30,in=150] (-1+\trans,3) ;
\draw [->,thick]  (0,-3) to [out=-30,in=-150] (-1+\trans,-3);
\fill (4,3.5) node {$y=Q(x)$};
\fill (4,-3.5) node {$x=Q^{-1}(y)$};
\fill (\trans*1.3-2,-2.75) node {$bd(\mathcal{X})$};

\end{tikzpicture}
\caption{From payoffs to strategies and back, the non-steep case:
deformation of neighborhoods under the Euclidean projection map $\mirror = \Eucl$.}
\end{figure}


\begin{tcolorbox}[enhanced,width=5in, drop fuzzy shadow southwest, 
                    boxrule=0.4pt,sharp corners,colframe=yellow!80!black,colback=yellow!10]

The following lemma shows that, for generic games, if the underlying regularizer is non-steep, all trajectories starting near an asymptotically stable point $\eq$ attain the face of $\eq$ in some uniform, finite time.
The intuition for this is that, generically, for any player $\play$, the coordinates of $\dpoint_\play$ that correspond to the support of $\eq_\play$ increase with a ``speed'' that is uniformly higher than those strategies not supported in $\eq$.
The regularizer of player $\play$ could possibly act in favor of the coordinates that do not belong to the support, but in a bounded way, since it is non-steep.
Thus, there is a time after which the coordinates of $\dpoint_\play$ corresponding to the support are bigger enough than the other coordinates, so that the mirror map $\mirror_\play$ keeps returning a point with support equal to the support of $\eq_\play$.
%
\end{tcolorbox}

\begin{lemma}\label{lem:finite_converg_Nonsteep}
Let $\eq$ be an asymptotically stable equilirium of  a generic finite game $\fingame$, with the regularizers used, being non-steep. For any neighborhood $U$ of $\eq$, there exists a   
neighborhood $U_0$ of $\eq$ and a finite time $T_0$ such that if $\point(t)=\mirror(\dpoint(t))$ is an orbit of \ac{FTRL} starting at $\point(0)\in U_0$, then $\supp(\point(t))=\supp(\eq)$ for all $t\geq T_0$.
\end{lemma}

\begin{proof}
By the genericity assumption, all Nash equilibria are quasi-strict. Clearly we have that for any player  $\play$, $u_{\play\pure}(\eq)>u_{\play\purealt}(\eq)$ for all $\pure \in \supp(\eq_\play)\equiv \pures^*_\play$ and $\purealt_\play \notin \supp(\eq_\play)$. Thus, by continuity there exists some neighborhood $U$ of $\eq$ and a $c>0$ so that for any $\point \in U$ and  any player $\play$, $u_{\play\pure}(\point)> u_{\play\purealt}(\point)+c$ for all $\pure \in \pures^*_\play, \purealt\in\pures_\play\setminus\pures^*_\play$. Additionally, we can choose $U$ small enough  so that for all $\point\in U$, $\supp(\eq)\subseteq\supp(\point)$. Since $\eq$ is asymptotically stable there exists a   neighborhood $U_0$ of $\eq$ so that $\point(t)\in U$ for all $t$ whenever $\point(0)\in U_0$,  and $\lim_{t\to\infty}\point(t)=\eq$.

Consider some $\point(0)\in U_0$ and let $\pure \in \pures^*$, $\purealt\in \pures\setminus\pures^*$. For ease of notation in the following we focus on the perspective of an arbitrary player, say $\play$, and  omit $\play$ from the subscripts. 

By Lemma \ref{lem:KKT}, for any $t\geq0$ there exist a $\mu(t)$ and non negative $v_\pure(t)$'s so that 
\begin{align*}
 \dpoint_\pure(t)&=g_\pure(\point(t))+\mu(t)&\forall \pure \in \pures^*\\  
 \dpoint_\pure(t)&=g_\pure(\point(t))+\mu(t)-v_\purealt(t)&\forall \purealt \in \pures\setminus\pures^* 
\end{align*}
since, by complementary slackness, $v_\pure(t)=0$, whenever $x_\pure(t)>0$.
Subtracting we get 
\begin{equation}\label{eq:Gbound}
\dpoint_\pure(t)-\dpoint_\purealt(t)=g_\pure(\point(t))-g_\purealt(\point(t))+v_\purealt(t)\leq v_\purealt(t)+G
\end{equation}
 with the inequality following, for some constant $G$, by $\hreg$ being non-steep. 

On the other hand by the definition of the dynamics, using  \Cref{eq:Gbound} and that 
$u_\pure(\point(t))> u_\purealt(\point(t))+c$, for all $t$ (since $x(0)\in U_0$), we get
\begin{align*}
    \dpoint_\pure(t)-\dpoint_\purealt(t)&=\dpoint_\pure(0)-\dpoint_\purealt(0))+\int_0^t[u_\pure(\point(s))-u_\purealt(\point(s))]ds\\
    &> g_\pure(\point(0))-g_\purealt(\point(0))+v_\purealt(0)+ct\\
    &\geq ct+ v_\purealt(0)-G
\end{align*}
with the last inequality following again by $\hreg$ being non-steep.
Combining the latter with \Cref{eq:Gbound}, and since $v_\purealt(0)\geq0$, we get 
$$v_\purealt(t)+G> ct+ v_\purealt(0)-G\Rightarrow v_\purealt(t)> ct-2G$$
 which implies that for $t\geq\frac{2G}{c}$ it is $v_\purealt(t)>0$. This in turn, by complementary slackness, yields $\point_\purealt(t)=0$ for all $t\geq\frac{2G}{c}$, implying $\supp(\point(t))\subseteq\supp(\eq)$. By the choice of $U$ and since $\forall t:\point(t)\in U$ we have $\supp(\point(t))\supseteq\supp(\eq)$  and thus setting $T_0=\frac{2G}{c}$ proves the claim, since the above holds for any player $\play$. 
\end{proof}

\begin{tcolorbox}[enhanced,width=5in, drop fuzzy shadow southwest, 
                    boxrule=0.4pt,sharp corners,colframe=yellow!80!black,colback=yellow!10]

The main result of this section is the following theorem that covers the non-steep case, stating that for generic games  at an asymptotically stable point under non-steep regularizers, every player plays a pure strategy.
The proof combines  results presented above. It reaches a contradiction by showing that points in a small enough neighborhood of the asymptotically stable point $\eq$, instead of converging to it as they ought to, they follow recurrent trajectories.  In a first step it finds points  that after a finite time $T_0$ reach and stay forever at the simplex formed by the support of $\eq$  (using \Cref{lem:finite_converg_Nonsteep}), which moreover have non-zero volume in that simplex. But then, these points follow \ac{FTRL} trajectories in the restricted simplex (\Cref{prop:FTRL-strat}) and $\eq$ belongs in the interior of this simplex.
However, we already know this cannot be the case in the interior (\Cref{thm:No_Aspoint_Interior}) and we follow a similar reasoning. 
\end{tcolorbox}

\begin{theorem}
If $\eq\in\points$ is an asymptotically stable point under non-steep regularizers of a generic game $\fingame$, then it consists of only 
pure strategies.
\end{theorem}

\begin{proof}
Let $\eq\in\points$ be asymptotically stable, $\pures^*=\supp(\eq)$, with  $|\pures_\play^*|=|\supp(\point_\play^*)|\geq 2$ for some player $\play$,
and  $\points^*$ be its respective simplex. Since $\eq$ is attracting, there exists some (bounded) neighborhood $U$ of $\eq$ for which if $\point(0)\in U$, then $\lim_{t\to \infty}\point(t)=\eq$.
 By \Cref{lem:finite_converg_Nonsteep}, there exists a   
 neighborhood  $U_0$ of $\eq$ and a finite time $T_0$ such that if $\point(t)=\mirror(\dpoint(t))$ is an orbit of \ac{FTRL} starting at $\point(0)\in U$, then $\supp(\point(t))=\pures^*$ for all $t\geq T_0$. 
 
By \Cref{prop:FTRL-strat}, for $t\geq T_0$ all trajectories satisfy  \Cref{eq:FTRL-strat} and these trajectories coincide with the trajectories of a generic game $\fingame'$ played on $\pures^*$, with the restricted simplex being $\strats^*$. 
At the same time, similar to the proof of  \Cref{thm:No_Asset_Interior},  $\flowmap(U_0,T_0)$ is a bounded (as a subset of  $U\intersect\strats^*$), positive measure for $\points^*$ (as the evolution of an  open set after a finite time (recall $|\pures_\play^*|\geq2$ for some $\play$)), and invariant (\Cref{lem:X_volume_preserving}) set of $\strats^*$ and, therefore, it is also Poincaré recurrent, contradicting that for $x(0)\in\flowmap(U_0,T_0)\subseteq U $, $\lim_{t\to \infty}\point(t)=\eq$, as implied by the asymptotic stability of $\eq$. 
Thus, if $\eq\in\points$ is asymptotically stable then  $|\supp(\point_\play^*)|=1$ for all $\play$.
 %
\end{proof}

%

As we stated in the beginning of this appendix, the steep case of \cref{thm:unstable-mixed} comes as a corollary (\cref{cor:steepUnst}) of a more general result on asymptotically stable sets that we prove in the next section (\cref{thm:unstable-set}).

\section{Proof of \cref{thm:unstable-set} and \cref{thm:unstable-mixed}: the steep case}
\label{app:steep}

\begin{tcolorbox}[enhanced,width=5in, drop fuzzy shadow southwest, 
                    boxrule=0.4pt,sharp corners,colframe=yellow!80!black,colback=yellow!10]

The following theorem shows that any asymptotically stable set $A$ cannot be contained in the interior of any non-singleton face $\points'$. This comes as a consequence of \cref{thm:No_Asset_Interior}. When the regularizers are steep, any point starting in $\intstratsPr$ stays in $\intstratsPr$ over time and $A$ being asymptotically stable implies that $A\intersect \points'$ is an asymptotically stable set under the \ac{FTRL} dynamics of the  restricted game played on $\points'$. But for the restricted game \cref{thm:No_Asset_Interior} applies excluding the possibility of $A\intersect \points'$ being an asymptotically stable set inside $\intstratsPr$. 
\end{tcolorbox}

\begin{theorem}\label{thm:steep_set_not_Interior}
Let $A\subseteq \points$ be an asymptotically stable set intersecting a non-singleton face $\points'$ of $\points$. Then, $A\intersect\points'$ cannot be contained in the relative interior of $\points'$.
\end{theorem}
\begin{proof}
To reach a contradiction let $A$  intersect a non-singleton face $\points'$ of $\points$ and $A' = A \intersect \points'$ be a subset of the relative interior of $\points'$. We will show that if this is the case then $A'$ is an asymptotically stable set under the dynamics of \ac{FTRL} restricted to $\points'$,  that lies in the relative  interior of  $\points'$. This  contradicts Theorem \ref{thm:No_Asset_Interior}.

To reach the contradiction, we go on to prove that $A'$ is an asymptotically stable set under \ac{FTRL} dynamics in $\points'$. We will crucially use that with steep regularizers  for any $\point(0)\in\intstratsPr$, $\point(t)\in \intstratsPr$, for all $t\geq0$, i.e.,  $\points'$ is forward invariant under \ac{FTRL}  (\cref{prop:FTRL-strat}).

To show Lyapunov stability of $A'$ in $\points'$ pick any neighborhood $U'$ of $A'$ in $\points'$. It  can be written as $U'=U\intersect \points'$ for some  neighborhood $U$ of $A$ in $\points$. Since $A$ is Lyapunov stable, there exists  a neighborhood $U_0$ of $A$ in $\points$ such that for any $\point(0)\in U_0$, it is  $\point(t)\in U$ for all $t\geq 0$. Let $U'_0=U_0\intersect \points'$. Using that $\points'$ is forward invariant, the latter implies that for any $\point(0)\in U'_0=U_0\intersect \points'$, it is  $\point(t)\in U\intersect\points'=U'$ for all $t\geq 0$, as needed.

We use similar ideas to show that $A'$ is attracting in $\points'$. Since $A$ is attracting in $\points$ there exist a neighborhood $U$ of  $A$ in $\points$ such that for any $\point(0)\in U$, $\point(t)\to A$. Let $U'=U\intersect \points'$. The latter combined with the forward invariance of $\points'$ implies that  for any $\point(0)\in U\intersect \points'=U'$, $\point(t)\to A\intersect \points'= A'$, as needed. 
 %
\end{proof}

\begin{tcolorbox}[enhanced,width=5in, drop fuzzy shadow southwest, 
                    boxrule=0.4pt,sharp corners,colframe=yellow!80!black,colback=yellow!10]

As a corollary of the above theorem we may get \cref{thm:unstable-set}, restated below. Whenever an asymptotically stable set intersects a non-singleton face it must intersect its respective boundary, and thus a face of smaller dimension. Consequently, ``in the long run'', it must intersect a singleton face. Put differently, it should contain a point consisting of only pure strategies.
\end{tcolorbox}

\stableSet*

\begin{proof}
Let $A$ be asymptotically stable and $\points_{min}$ be a face of minimal dimension intersected by $A$ which is  not a singleton. By Theorem \ref{thm:steep_set_not_Interior}, $A$ cannot be contained in the relative interior of $\points_{min}$, so it must intersect the boundary of $\points_{min}$. However, this means that $A$ intersects a face of dimension strictly smaller than that of $\points_{min}$, a contradiction. Thus, $\points_{min}$ is a singleton and  $A$ must contain a vertex of $\points$.
\end{proof}

\begin{corollary}\label{cor:steepUnst}
If $\point$ is an asymptotically stable point,  then it consists of  only pure strategies.
\end{corollary}

\bibliographystyle{plainfull}
\bibliography{IEEEabrv,bibtex/Bibliography-PM,bibtex/refer}

\end{document}